\documentclass[11pt]{article}
\usepackage{amsmath, amssymb, amscd, amsthm, amsfonts}
\usepackage{graphicx}
\usepackage{hyperref}
\usepackage{todonotes}
\usepackage[ruled,vlined,noresetcount]{algorithm2e}

\usepackage{setspace}
\usepackage{float}
\usepackage{blindtext}
\usepackage{authblk}
\usepackage[square,numbers]{natbib}
\usepackage[letterpaper,margin=1.1in]{geometry}

\title{A Lightweight Approach for State Machine Replication}

\author{
\begin{minipage}[t]{.48\linewidth}\centering
Christian Cachin\\\small\texttt{christian.cachin@unibe.ch}
\end{minipage}\hfill
\begin{minipage}[t]{.48\linewidth}\centering
Jinfeng Dou\\\small\texttt{jfdou@mail.upb.de}\vspace*{3mm}
\end{minipage}\\
\begin{minipage}[t]{.48\linewidth}\centering
Christian Scheideler\\\small\texttt{scheideler@upb.de}
\end{minipage}\hfill
\begin{minipage}[t]{.48\linewidth}\centering
Philipp Schneider\thanks{Philipp Schneider was partially supported by a grant from Avalanche, Inc.\ to the University of Bern.}\\
\small\texttt{philipp.schneider@cispa.de}
\end{minipage}
}

\date{}

\theoremstyle{definition}
\newtheorem{theorem}{Theorem}
\newtheorem{lemma}{Lemma}[section]
\newtheorem{definition}[lemma]{Definition}

\newtheorem{corollary}[lemma]{Corollary}

\newcommand{\bigO}{\smash{\ensuremath{O}}}

\newtheorem{claim}{Claim}

\def\E{{\rm E}}
\def\V{{\rm V}}

\begin{document}

\maketitle

\begin{abstract}
We present a lightweight solution for state machine replication with commitment certificates. Specifically, we adapt and analyze a median rule for the stabilizing consensus problem \cite{DoerrGMSS11} to operate in a client-server setting where arbitrary servers may be blocked adaptively based on past system information. We further extend our protocol by compressing information about committed commands, thus keeping the protocol lightweight, while still enabling clients to easily prove that their commands have indeed been committed on the shared state.
Our approach guarantees liveness as long as at most a constant fraction of servers are blocked, ensures safety under any number of blocked servers, and supports fast recovery even after all servers are blocked. In addition to offering near-optimal asymptotic performance in several respects, our method is fully decentralized, unlike other near-optimal solutions that rely on leaders. In particular, our solution is robust against adversaries that target key servers (which captures insider-based denial-of-service attacks), whereas leader-based approaches fail under such a blocking model. 
\end{abstract}


\newpage

\tableofcontents

\newpage

\section{Introduction}
In the \emph{state machine replication} (SMR) problem \cite{Schneider90}, a set of servers is required to replicate a state machine in a consistent way by agreeing on an order in which commands from clients are to be committed on it. Formally, a \emph{state machine} $(\mathcal{S},\mathcal{C})$ consists of a \emph{state space} $\mathcal{S}$ and a collection of \emph{commands} $\mathcal{C}$ acting on that state space. Each server maintains a copy of the state machine (simply called a \emph{shared state} in the following). The clients may send their commands to any of the servers. 
The goal is to satisfy the following conditions:
\begin{itemize}
\item Safety: At any point in time, for any pair of correct servers, the sequence of committed commands on one server is a prefix of the sequence on the other.
\item Liveness: Every command that was submitted to a correct server is eventually committed on all correct servers.
\end{itemize}
Safety and liveness imply that if every server executes the commands in the order in which they were committed by that server, the sequence of transitions on the shared state will be the same for each correct server, and any transition at one correct server will eventually be performed by all other correct servers.
An important application of SMR is a system that manages financial transactions.
However, for financial transactions, it is not only important to maintain a shared state (i.e., the account balances), but also records of the transactions. 
In cryptocurrencies, transactions are typically recorded on a large blockchain that is replicated across many servers. For example, Bitcoin’s ledger exceeds 699 GB and is maintained by thousands of full nodes \cite{BBS}. 
A much more space-efficient alternative for servers would be to maintain only the shared state and shift the burden of maintaining a proof that a transaction took place to the originating clients, who are naturally incentivized to do so, to be able to prove successful commitment of their transactions to third parties.

We present a scalable and robust solution to the SMR problem that provides safety and liveness under a blocking adversary and minimizes the amount of information that must be maintained by the servers by offloading the responsibility of certifying the commitment of commands to the clients. Specifically, we require clients to certify commitment of their commands to other clients, \textit{without} using long-term cryptographic signatures, where servers may act only as verifiers that retain a small amount of metadata but not the complete history of committed commands.
This poses several intertwined challenges, especially under failures. We adopt a \emph{blocking-failure} model under an adaptive adversary that has full information with a one-round delay. 
Adaptive blocking failures subsume both crash-stop and crash–recovery behavior since such an adversary may permanently or temporarily block servers.
The core difficulty is that the adversary can potentially isolate groups of servers from each other by alternate blockings, making it hard to maintain safety.

Unlike many contemporary solutions for fault-tolerant SMR, which use leaders that can be easily suppressed by our adaptive adversary, our protocol is leaderless and, to our knowledge, the first to realize SMR via the median-rule consensus dynamics \cite{DoerrGMSS11}. We prove safety and liveness properties of our protocol in scenarios with increasing difficulty, for which we provide a formal framework by defining a sequence of core problems built on top of each other. 
This work focuses on the \emph{theoretical} foundations of SMR under a blocking adversary based on the median rule by providing an analysis of the consensus dynamics in the style of \cite{BecchettiCN20,DoerrGMSS11} and using these insights to obtain a light-weight protocol that is robust to blocking failures.

\subsection{The Model} 

\paragraph*{Client-Server Model} We distinguish between \emph{clients} and \emph{servers}: clients issue commands, while servers maintain a copy of the state machine. We assume a fixed set of servers $\{1, \dots, n\}$. The set of clients may change arbitrarily over time. 
We focus on minimizing communication overhead, latency, and storage, and therefore, any local computations are assumed to be negligible (which is indeed the case in our solution).

\paragraph*{Network Model} Communication between clients and servers is done by point-to-point message passing. While clients can be addressed directly, for servers we assume the weaker gossip communication model (see, e.g., \cite{DBLP:conf/podc/HaeuplerMS18}) where servers can exchange a point-to-point message with anyone but \textit{cannot} choose whom to contact. Messages are exchanged using either a \emph{push} operation, which sends a message to a random server, or a \emph{pull} operation, which asks a random server to reply to the sender, where every server has the same probability to be chosen as the target of the operation. Gossip-based protocols are highly practical due to their inherent load balancing, rapid convergence, simplicity, and robustness under stress or disruption. For instance, they play a key role in Avalanche consensus \cite{AmoresSesarCS24,abs-1906-08936} and anonymous networks like Tor. Because the gossip model is used for the servers, neither servers nor clients need to know the server identifiers or the exact number $n$; having some polynomial estimate of $n$ is sufficient for our protocols to work.

\paragraph*{Synchrony Assumptions}
In \textit{synchronous message passing}, communication takes place in discrete time steps called \emph{rounds}, and any push or pull operation initiated at the beginning of some round is completed before the next round starts. By contrast, \textit{asynchronous message passing} provides no timing guarantees: messages may get arbitrarily delayed. In this work, we assume synchronous message passing for server-to-server communication and asynchronous message passing for interactions between clients and servers. 

\paragraph*{Failure Model}
We assume that the clients cannot impersonate other clients but may otherwise behave in a Byzantine manner. Servers are subject to blocking failures that are adversarially triggered but otherwise correctly follow the given protocol.
The blocking adversary can isolate any server at any point in time, preventing it from sending or receiving messages in a given round. A server does not know whether it is isolated or not, and in case of a failed push or pull request, no information is given on who of the two involved parties caused the failure. This covers real-world phenomena such as denial-of-service attacks or network failures. It also captures aspects of network asynchrony or unreliable links, in the sense that blocked parts of the network experience delays or message loss. Because blocked servers are effectively cut out of the protocol, this model also subsumes the crash-recovery model, with the adversary controlling the timing of crashes and subsequent recovery. 
We consider a \emph{late} adversary that adapts its blocking decisions based on previously observed system states (e.g.\ \cite{DBLP:conf/opodis/AhmadiGKM15,RobinsonSS18}) which allows us to rely on randomness to achieve network-wide outcomes that are outside the adversary's control. Lateness is a natural assumption, as information about the servers first has to be aggregated and evaluated by an adversary in order to target its attack.

\begin{definition}[$\alpha$-late, $\beta$-blocking adversary]
\label{def:alpha_late_beta_blocking}
Let $\alpha \in \mathbb{N}_0$ and $0 \leq \beta \leq 1$ be fixed parameters.
\begin{itemize}
    \item In round $r$, the \emph{$\alpha$-late} adversary only knows the local states (including random choices) of the servers up to the beginning of round $r \! - \! \alpha$. 
    \item In each round, the $\beta$-\emph{blocking} adversary can block up to a $\beta$-fraction of servers. For simplicity, we assume that a server will either be completely blocked or completely unblocked during a round. A blocked server cannot send or receive any messages in that round.
\end{itemize}
\end{definition}

We assume a 1-late blocking adversary ($\alpha=1$): in round $t$ it knows the entire system state only up to the beginning of round $t-1$. This model is incomparable to, and not weaker than, the classical Byzantine model. Byzantine-tolerant protocols typically posit a fixed, sufficiently large correct majority that is always available, whereas here the adversary may adaptively block any subset of servers at any time. In particular, leader-based consensus protocols cannot maintain liveness against such adaptive blocking.
More generally, for an $\alpha$-late adversary with $\alpha \ge 1$, the adversary cannot observe which client requests arrive in the current round, nor the recipients of the servers’ push/pull messages from the previous round, and therefore, it cannot reliably infer the servers’ current states. This one-round opacity is crucial for fast commitment, because with $\Theta(n)$ adaptive crash failures under a \textit{full}-information adversary, any synchronous randomized consensus protocol needs $\tilde{\Theta}(\!\sqrt{n})$ rounds \cite{Bar-JosephB98}.

\paragraph*{Probabilistic Concepts}
We use the standard concept of  \emph{with high probability} (w.h.p.) where the total probability of algorithm failure is upper bounded by $1/n^c$ for \textit{any} constant $c > 0$. This ensures that the likelihood that our algorithm does not satisfy the safety and liveness vanishes for sufficiently large $n$. We further discuss how to deal with the unlikely event of a single failure in Section \ref{sec:extension}. The probabilistic concepts we frequently use are given in Appendix \ref{sec:probabilistic_concepts}.

\subsection{Overview of our Approach} \label{sec:overview}

Our goal is a robust protocol for state-machine replication (SMR) that is lightweight in terms of memory usage, allows signature-free certification of commitment of commands, and is highly scalable in terms of communication overhead and latency. Specific objectives are:
\begin{itemize}
    \item \textbf{Small latency} from receiving a client-issued command by some server to the commitment of the command on all non-blocked servers.
    \item \textbf{Small communication overhead} 
    defined as the additional amount of information the server has to send and receive per command compared to receiving that command once.
    \item \textbf{Small memory requirement per server}, where a server only has to store the shared state, the non-committed commands and a small amount of additional information.
    \item \textbf{Small memory requirement per client}, where a client only has to store small certificates for each of its own committed commands.
    \item \textbf{Liveness} against a 1-late adversary blocking a constant fraction of servers ($\Theta(1)$-blocking).
    \item \textbf{Safety} even against a 1-late adversary potentially blocking all servers (1-blocking).
    \item \textbf{Fast recovery} after any period of time with a 1-late 1-blocking adversary. 
\end{itemize}
Our solution also avoids reliance on cryptographic signatures, which is motivated by the fact that signatures cause computational overhead, might expire, keys can get stolen, or even be broken at some point (if certain cryptographic assumptions turn out to be false). Hence, small commitment certificates cannot be simple signed statements by the server. Our solution only assumes a publicly known collision-resistant hash function, i.e., a function $h$ for which it is computationally hard to find $x \ne y$ with $h(x)=h(y)$.

In the following, we formally describe the problems that we are tackling towards a solution satisfying the criteria above. We develop this solution step-by-step in a sense that a solution of a simpler problem is used for the solution to the next, more complex problem. We start with a simple distributed algorithm for the \emph{stabilizing consensus problem}, which adapts the classic consensus problem to capture blocked servers. We call servers \textit{useful} if they are currently non-blocked and store a value.

\begin{definition}[Stabilizing Consensus Problem]
    \label{def:stabilizing_consensus}
    Every server holds an arbitrary initial value.
    A protocol solves the stabilizing consensus problem if it satisfies:
    \begin{itemize}
        \item {\bf Agreement}: There exists a round after which all useful servers hold the same value.
        \item {\bf Availability}: In each round, at least a constant fraction of servers is useful.
        \item {\bf Validity}: If a server stores some value $x$ at the end of a round, then some server must have stored $x$ at the beginning of that round.
    \end{itemize}
\end{definition}

Note that a protocol ensuring agreement should eliminate values from blocked servers to prevent the survival of outdated ones. Our agreement requirement for \emph{useful} servers is analogous to Byzantine agreement (BA), where only \emph{correct} servers must agree. The key difference is that the set of useful servers changes over time, due to the adaptive adversary and the randomness of gossip. A main challenge, therefore, is to ensure that a constant fraction of servers remains useful (available) at all times despite this shifting membership. We address this problem in Section \ref{sec:median} by combining a simple median rule in \cite{DoerrGMSS11} with undecided-state dynamics (see, e.g., \cite{AngluinAE08}). In this median rule, each server in each round requests the value from a random, constant-size subset of other servers, adopts the median value if enough values are received, and otherwise deletes its value, thereby becoming useless (i.e., not useful).

We then adapt this algorithm to solve the \emph{stabilizing SMR problem} in Section~\ref{sec:xmedian}, whose goal is to reach agreement on an ordered sequence of client-issued commands called \emph{log}, even under an adversary that may block servers. Here, we say a command is \emph{injected} as soon as it is first received from the client by a non-blocked server and assume that every command is unique (so that a client command does not appear twice
in a log by mistake). Analogously to the single value case, we call non-blocked servers with a non-empty log \textit{useful}.

\begin{definition}[Stabilizing SMR Problem]
    Each server starts with a log containing a seed command (so all servers are initially useful). A protocol solves the problem if it satisfies:
    \label{def:stabilizing_smr}
    \begin{itemize}
        \item {\bf Agreement}: For every injected command $x$, there exists a round after which all useful servers contain $x$ at the same position in their logs.
        \item {\bf Availability}: At least a constant fraction of servers are useful.
        \item {\bf Validity}: If a server has some command $x$ in its log at the end of a round, then it was either injected by a client or some server had $x$ in its log at the beginning of that round.
    \end{itemize}
\end{definition}

We show in Section \ref{sec:xmedian} that the median rule can be adapted to work in the context of SMR by using the following idea: In each round, each server requests the logs from a random, constant-size subset of other servers. If enough logs are received, it adopts the median log (according to lexicographic order), appending all previously seen commands that are not contained in the median log to the end, and otherwise becomes useless (deletes its log). The \emph{latency} of an algorithm that solves this problem is the number of rounds from the injection until the agreement property is satisfied. 

Our solution in Section \ref{sec:xmedian} can already be employed to implement state machine replication with small latency under a blocking adversary. However, this has the obvious drawback that entire logs (that can grow arbitrarily large) must be memorized and exchanged. Therefore, as next step in our iteration, we devise locally checkable rules for servers to determine when commands can be safely \emph{committed} in Section 
\ref{sec:cmedian} and Section \ref{sec:proofs}. Committed commands can be executed on the local state (by order of time of commitment) and subsequently forgotten, with the guarantee that any other non-blocked server will commit the same commands in the same order. We encapsulate this in the ``commitment problem'', which mirrors the standard SMR problem to our setting:

\begin{definition}[Commitment Problem]
    \label{def:commitment_problem}
    A protocol solves the problem if it satisfies:
    \phantom{a}
    \begin{itemize}
        \item {\bf Strong Safety}: Any two useful servers have the same sequence of committed commands.
        \item {\bf Availability}: In each round, at least a constant fraction of servers is useful.
        \item {\bf Liveness}: Every injected command is eventually committed by all useful servers.
    \end{itemize}
\end{definition}

When satisfying the criteria above, useful servers can execute commands on the shared state whenever they consider them to be committed and forget about them afterwards. This way, useful servers only need to maintain a shared state and a log of those commands that are not yet committed, leading to a significant reduction in the size of the messages as well as the storage needed by the servers. Servers that have been useless for a long period of time (e.g., due to frequent adversary blockings), will eventually be updated with the shared state by useful servers when they are able to successfully contact them again. 

The clear drawback of deleting committed commands is that it erases the history of commands (e.g., transactions in a blockchain). This is the next problem we address: preserving the history without the servers needing to store it. Our solution, given in Section \ref{sec:proofs}, is to offload the main burden of storing committed commands to the clients that issued them together with a small certificate that allows clients to prove to any third party that their command is indeed part of the committed history, whereas the servers only store a small amount of metadata to verify certificates. Note that our approach makes the client completely responsible for storing its issued commands and the associated certificate. However, loss of this data by one client never affects certificates of others.

Finally, we address in Section \ref{sec:recovery} the problem of recovering from any number of blocked servers over any period of time. A central challenge that arises from such an adversary is that the number of useful servers might be driven to zero. In our previous approach, we rely on responses from useful servers in order to transition previously useless servers to useful again, which fails in this setting. Worse still, an adversary may keep the number of useful servers at a small but changing minority, which poses the risk of disagreement on the sequence of committed commands.
To address this, we make use of two additional protocols on top of our solution for the Commitment Problem: One that maintains the most recent shared state and non-committed commands (checkpoint protocol) and one that decides whether a reset to a prior checkpoint is required (reset protocol). For the first one, the servers take checkpoints of the shared state and disseminate these to agree on a most recent one. For the second one, the servers use a consensus mechanism to decide whether to reset the system or not. Formally, our protocol solves the following problem.

\begin{definition}[Recovery Problem]
    \label{def:recovery_problem}
    A protocol solves the recovery problem if it satisfies:
    \begin{itemize}
        \item {\bf Monotonicity}: 
        Under any blocking adversary, for every server $i$ and all $t<t'$, the committed sequence at time $t$ is a prefix of the committed sequence at time $t'$.
        \item {\bf Recovery}: Once the blocking adversary is below its intended threshold, eventually strong safety, availability and liveness of Definition \ref{def:commitment_problem} hold again.   
    \end{itemize}
\end{definition}

Monotonicity ensures that no server starts an ``alternate'' commitment sequence; a situation akin to double spending in blockchains. Recovery guarantees that if any server has committed a command, it will eventually appear in every useful server's committed sequence once the adversary reverts to a bounded attack.

\subparagraph*{Our Contributions}
Technically, we combine the median rule of \cite{DoerrGMSS11} with undecided-state dynamics (e.g., \cite{AngluinAE08}) and extend it to SMR with lightweight servers. To our knowledge, this is the first use of median-rule dynamics for SMR and the first analysis under a late blocking adversary. Prior work largely targets Byzantine faults and has either higher overheads or is leader-based and thus cannot ensure liveness against our adaptive blocking adversary.

Our main result is a protocol that meets all performance criteria listed at the beginning of Section~\ref{sec:overview}. Concretely, it achieves logarithmic (in $n$) communication per client command and logarithmic latency. Further, server space is logarithmic in the length of the committed history plus logarithmic space per client. Clients only store logarithmic information for each command they issue. Our protocol guarantees {liveness} against any 1-late, $1/10$-blocking adversary and {strong safety} (Definition~\ref{def:commitment_problem}) against any 1-late adversary, with no bound on the number of blocked servers. Finally, it provides logarithmic-time \emph{recovery} from any number of blocked servers under a 1-late adversary (Definition~\ref{def:recovery_problem}). Our protocol compares favorably with other gossip-based blockchain protocols, in particular Avalanche, which achieves the same performance but can only provide liveness against $O(\!\sqrt n)$ blocking failures (whereas we tolerate $\Omega(n)$) and does not consider the ramifications of keeping servers lightweight.

Further, we show that servers do not need to store the full command history, as is common in blockchains. Instead, commitment certificates are offloaded to the issuing clients. Servers keep only the shared state and the pool of uncommitted commands. Clients need certificate updates only when they submit a new command, and can remain offline otherwise. Our approach has additional benefits to privacy: after a command is committed, only the client will retain command-specific information.
Technically, the recovery protocol is the most delicate component. Our techniques require new insights into the median rule under blockings and two tightly coupled subprotocols, enabling rapid (logarithmic-time) recovery after arbitrary blocking attacks.

\subsection{Related Work}

\paragraph{Leader and DAG Based Consensus Protocols}
Most practical Byzantine fault-tolerant (BFT) consensus protocols follow either a leader-based or a DAG-based approach. Examples of leader-based protocols are PBFT \cite{DBLP:conf/osdi/CastroL99/PBFT}, HotStuff \cite{YinMRGA19}, Autobahn \cite{DBLP:conf/sosp/GiridharanSAAC24}, and ProBFT \cite{DBLP:conf/podc/AvelasHADB24}, and examples of DAG-based protocols are DAG-Rider \cite{DBLP:conf/podc/KeidarKNS21}, Narwhal \cite{DBLP:conf/eurosys/DanezisKSS22}, Bullshark \cite{SpiegelmanGSK22}, and Shoal++ \cite{Shoal++}.
Most of the BFT protocols can tolerate large-scale Byzantine behavior in the network. However, with the advent of trusted execution environments (TEEs), such strong fault tolerance might not be strictly necessary: adversarial peers can block a trusted device or obstruct message delivery, but cannot tamper with its internal execution. This has led to efficient consensus protocols based on trusted components \cite{DBLP:conf/sosp/ChunMSK07,DBLP:conf/nsdi/LevinDLM09,DBLP:journals/tc/VeroneseCBLV13}, even before the blockchain era.

All of the works above assume that the set of participants under adversary control is static. Although we consider only attacks that block participants, the blocking can change \textit{adaptively} from round to round. This rules out leader-based approaches since such an adversary can simply block a leader once it has been elected with a
denial-of-service (DoS) attack.
Heavily concentrating communication on a single participant also leads to imbalances in communication load and raises fairness concerns. Although rotating the leader over time can address fairness, our protocol avoids leaders altogether by being fully decentralized. In addition, while many state-of-the-art protocols rely on batching, coding, or threshold signatures to achieve near-optimal message complexity, our approach achieves near-optimal performance \emph{without} such mechanisms.

The primary advantage of DAG-based protocols is liveness under full asynchrony, whereas our protocol offers liveness only under partial synchrony, when enough participants are synchronous. However, DAG-based approaches might make very slow progress: a classical lower bound for randomized consensus in the asynchronous model shows that, for any integer $k$, an $f$-resilient algorithm over $n$ processes fails to terminate within $k(n-f)$ steps with probability at least $1/c^k$ (for some constant $c$) \cite{AttiyaH10}. Thus, asynchrony must persist for a long time before one can benefit from \textit{any} protocol that offers liveness under asynchrony, whereas our protocol recovers quickly once a period of arbitrary asynchrony ends.

\paragraph{Gossip Based Consensus Protocols}
In contrast to leader-based or DAG-based protocols, our approach relies on the gossip communication model. Gossip-based protocols 
have been studied extensively for various tasks, including information dissemination \cite{DBLP:conf/focs/KarpSSV00}, aggregation \cite{DBLP:conf/podc/HaeuplerMS18}, network coding \cite{DBLP:conf/stoc/Haeupler11}, and consensus \cite{DoerrGMSS11}. See \cite{BecchettiCN20} for a survey on gossip-based protocols. Due to their fast convergence and resilience under stress and disruptions, they have also been employed in various
blockchain solutions, such as Bitcoin \cite{DBLP:conf/icdcn/CrucianiP23} and Tendermint \cite{DBLP:journals/corr/abs-1807-04938}, for network maintenance and information dissemination. Most closely aligned with our work are anti-entropy protocols \cite{BirmanHOXBM99,PetersenSTTD97,RenesseDGT08}, including the Snow consensus protocol used in Avalanche \cite{abs-1906-08936}. While many of these protocols have been analyzed under message loss and some under Byzantine behavior (e.g., \cite{AmoresSesarCS24,abs-1906-08936}), their efficiency and robustness are often evaluated primarily through simulations. In contrast, we provide a rigorous analysis, building on works such as \cite{DoerrGMSS11, RobinsonSS18}. For the gossip model, it is known that under $\Theta(n)$ adaptive crash failures by a full-information adversary, any synchronous randomized consensus protocol needs $\tilde{\Theta}(\sqrt{n})$ rounds \cite{Bar-JosephB98}. We bypass this lower bound by withholding new information from the adversary for one round.

\paragraph{Certification of Client Commands} Our approach to generate client certificates is based on a Merkle hash forest, which builds on the Merkle Mountain Range technique used in Ethereum and other blockchains \cite{PeterTodd12}. While methods to store a compressed representation of the Merkle hash forest on servers have been used before, the novel challenge we address is ensuring that clients only need to receive the necessary hashing information once, at the moment the commitment of a command is acknowledged by a server, yet are still able to prove to any useful server, at any later time, that their command was committed, even though the servers just store a compressed version of the Merkle hash forest.

\section{Consensus with a Median Rule} \label{sec:median}

We start with an analysis of our proposed mechanism for stabilizing consensus that combines the median rule in \cite{DoerrGMSS11} with undecided-state dynamics (e.g., \cite{AngluinAE08}). The specific algorithm of the $(k,l)$-median rule can be found in Algorithm $\ref{alg:basic}$. Roughly, each server requests $k$ values from random peers, and if it receives at least $\ell$ it adopts the median value as its own, else it adopts $\bot$ (corresponding to the undecided state). Since blocked servers will never receive any replies, their value will be equal to $\bot$ at the end of the round. The validity property (Definition \ref{def:stabilizing_consensus}) is trivially satisfied by the $(k,\ell)$-median rule. Thus, it remains to show agreement and availability. Throughout the paper, we will focus on the specific case that $k=6$ and $\ell=3$ since that turns out to be the most suitable choice of parameters for our goals.

\begin{algorithm}
\caption{The $(k,l)$-median rule}
\label{alg:basic}
\textbf{Preconditions:}
Let $k, \ell >1$ with $k \ge \ell$ and $\ell$ odd. We assume that initially every server $i$ stores an arbitrary value $x_i \in K$, where $K$ is a discrete space with a total order. 

\smallskip

\textbf{Each server $i$ does the following each round:}
\begin{itemize}
\item send $k$ value requests to servers chosen uniformly and independently at random 
\item if $x_i \not=\bot$ then for any value request received from some server $j$, send $x_i$ back to $j$ 
\item if at least $\ell$ replies are received, choose a subset of $\ell$ of these replies uniformly at random and set $x_i$ to the median of the values sent by these replies 
\item if less than $\ell$ replies are received, set $x_i:=\bot$
\end{itemize}
\end{algorithm}    

\subsection{Availability}
We call a server \emph{useful} in round $t$ if it stores a value at the beginning of $t$ and is non-blocked during round $t$, and otherwise it is called \emph{useless}. By definition, availability only holds as long as a constant fraction of servers is useful. We will show three results that provide important insights into availability when using the $(6,3)$-median rule. The proofs are relatively technical. Here, we give an intuitive understanding of the claims and the way we use them later on.

The first result shows at which point a so-called \emph{spiral-of-death} occurs, which describes the event where the number of useful servers (almost inevitably) goes to 0. In Algorithm \ref{alg:basic} a server becomes useless if it does not obtain enough responses from others. However, this means there is a threshold for the number of useful servers under which it is very unlikely for servers to sample enough useful others, so the servers are trapped in a negative feedback loop where ever fewer useful servers cause even more to become useless. Below the given threshold the number of useful server decays with double exponential speed in the number of rounds (i.e, is completed after $O(\log \log n)$ rounds) even if no server is actually blocked.

\begin{lemma} \label{lem:spiral}
If the initial fraction of useful servers is at most $1/3-\varepsilon$ for any constant $\varepsilon>0$ then even if no server is blocked, within $O(\log \log n)$ rounds no server will be useful anymore, w.h.p.
\end{lemma}

\begin{proof}
For every server $i$ and round $t$ let the binary random variable $X_{i,t}$ be 1 if and only if $i$ is useful at round $t$. Let $X_t = \sum_{i=1}^n X_{i,t}$ and $x_t = X_t/n$. Consider some fixed round $t$ with $X_t \le (1/3-\varepsilon)n$ and suppose that no server will be blocked in round $t$ and $t+1$. Then, for any server $i$,
\begin{align*}
    \Pr[X_{i,t+1}=1] & =  1 - [(1-x_t)^6 + \binom{6}{1}x_t(1-x_t)^5+\binom{6}{2}x_t^2(1-x_t)^4]\\
    & = -10x_t^6+36x_t^5-45x_t^4+20x_t^3
\end{align*}
Let $f(x)=-10x^6+36x^5-45x^4+20x^3$. Then we get:
\begin{align*}
E[x_{t+1}] &\ = \frac{1}{n} \sum_{i=1}^{n} E[X_{i,t+1}]
              = \frac{1}{n} \sum_{i=1}^{n} \Pr[X_{i,t+1} = 1] \\
           &\ = \frac{1}{n} \sum_{i=1}^{n} f(x_t)
            = f(x_t) \\
\end{align*}
Set $g(x)$ = $f(x)-3x^2$. As can be seen from Figure~\ref{fig:g(x)}, $g(x) \le 0$ for all $x \in [0,1]$, which implies that $E[x_{t+1}] = f(x_t) \le 3x_t^2$ and therefore, $\E[X_{t+1}] \le 3 X_t^2/n$ for all $x_t \in [0,1]$. Since the probability that a server is useful in round $t+1$ is independent of the other servers, we can use the Chernoff bounds to show that for $\delta_t =  \sqrt{ c n \ln n}/X_t$ for any constant $c>0$ so that $\delta \le 1$,
\begin{align*}
  \Pr[X_{t+1} \geq (1+\delta_t) 3X_t^2/n] & \le e^{-( \sqrt {c n \ln n}/X_t)^2 \cdot 3X_t^2/(3n)} \\
  & = e^{-c \ln n} = 1/n^c
\end{align*}
Since $\delta_t \le \delta := 1/\sqrt{\ln n}$ as long as $X_t \ge  \sqrt{c n} \ln n$, it holds that $X_{t+1} \le (1+\delta) 3X_t^2/n$ in this case, w.h.p. For simplicity, assume that $X_{t+1} \le (1+\delta) 3X_t^2/n$ is guaranteed. Given any $x_0 \in [0,1]$, we want to show by induction on round $t\ge 1$ that $x_{t} \leq ((1+\delta)3x_{0})^{2^t}/[(1+\delta)3]$. The base case $x_{1} \leq ((1+\delta)3x_{0})^{2}/[(1+\delta)3]$ directly follows from the fact that $X_{t+1} \le (1+\delta) 3X_t^2/n$. Thus, assume that $x_{t-1} \leq ((1+\delta)3x_{0})^{2^{t-1}}/[(1+\delta)3]$. Then it follows that $x_{t} \leq (1+\delta)3x_{t-1}^2 \leq (1+\delta)3[((1+\delta)3x_{0})^{2^{t-1}}/((1+\delta)3)]^2 = ((1+\delta)3x_{0})^{2^t}/[(1+\delta)3]$. If $x_0 \le 1/3-\varepsilon$ for some constant $\varepsilon>0$ then $x_t \le ((1+\delta)(1-3\varepsilon))^{2^{t}}/[(1+\delta)3]$, which implies that in just $t=O(\log \log n)$ rounds, $X_t \le  \sqrt{c n} \ln n$. From that point on it is easy to show via Chernoff bounds that $X_{t+1} = O(\ln^2 n)$, w.h.p., and after a constant number of additional rounds, no useful server remains w.h.p.\qedhere

\begin{figure}
	\centering
    		\begin{minipage}[b]{0.5\linewidth}
   		 	\includegraphics[width=1\textwidth]{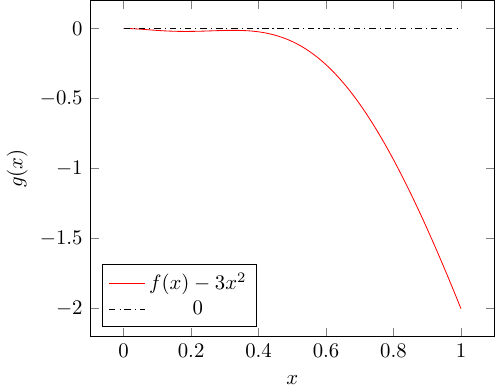}
    		\end{minipage}
	\caption{$g(x)=-10x^6+36x^5-45x^4+20x^3-3x^2$}
	\label{fig:g(x)}
\end{figure}
\end{proof}

On the positive side, there is a minimum fraction for the number of useful servers, where that fraction stays stable, even if the adversary blocks a certain fraction of servers. This will define our ``operating range'' where our system guarantees liveness.
Note that the specific thresholds of useful servers (3/4) and adversary ratio (1/10) are subject to our choice of parameters $(6,3)$ in the median rule. In the interest of keeping our (already extensive) analysis manageable, we settle on this set of parameters, as this turns out to have suitable properties for our purposes.

\begin{lemma} \label{lem:available}
If the initial fraction of useful servers is at least $1/2$ then for any 1-late $1/10$-blocking adversary, the fraction of useful servers monotonically converges to at least $3/4$ in just $O(\log n)$ rounds, w.h.p., and once it is at least $3/4$, it will be at least $3/4$, w.h.p., for any number of rounds that is polynomial in $n$. 
\end{lemma}

\begin{proof}
Let $X_{i,t}$, $X_t$, $x_t$ and $f(x)$ be defined as in the proof above. Note that $X_{i,t+1}=1$ if and only if server $i$ is not blocked in rounds $t$ and $t+1$ and it successfully applies the $(6,3)$-median rule in round $t$, which has a probability of $f(x_t)$. 

For any fixed non-blocked server $i$ in round $t$ let $S$ be the event that $i$ is successful in round $t$, i.e.,$x_i \not= \bot$ at the end of round $t$, and let $B$ be the event that $i$ will be blocked in round $t+1$. Since $\Pr[S]$ does not depend on the state of $i$ and every server $i$ performs the same random experiment to determine the new $x_i$, $\Pr[S]$ is the same for every non-blocked server. Moreover, because the adversary is 1-late, i.e., it does not know the random choices of the servers in round $t$ when deciding which servers to block in round $t+1$,
\[ 
  \Pr[S \mid B] = \Pr[S \mid \bar{B}] = \Pr[S] .
\]
Since, under a $1/10$-blocking adversary, the number of servers that are neither blocked in round $t$ nor in round $t+1$ is at least $4n/5$, the expected number of useful servers in round $t+1$ is at least $(4/5)n \cdot \Pr[S] = (4/5)f(x_t) \cdot n$. Set $g'(x) = (4/5)f(x) -(1+1/40)x$. As can be seen from Figure~\ref{fig:g'(x)}, $g'(x) \ge 0$ for all $x \in [1/2,3/4]$, which implies that $(4/5)f(x) \ge (1+1/40)x$ for all $x \in [1/2,3/4]$. Since the probability of a server to be successful is independent of the other servers, we can use Chernoff bounds to prove that for any $X_t \in [n/2,3n/4]$, $X_{t+1} \ge (1+\varepsilon)X_t$ for some constant $\varepsilon>0$, w.h.p. Once the fraction of useful servers is at least $3/4$, it will be at least $3/4$ in the next round, w.h.p., by a simple dominance argument (assume as a worst case that it is exactly $3/4$), which completes the proof.\qedhere

\begin{figure}
	\centering
    		\begin{minipage}[b]{0.5\linewidth}
   		 	\includegraphics[width=1\textwidth]{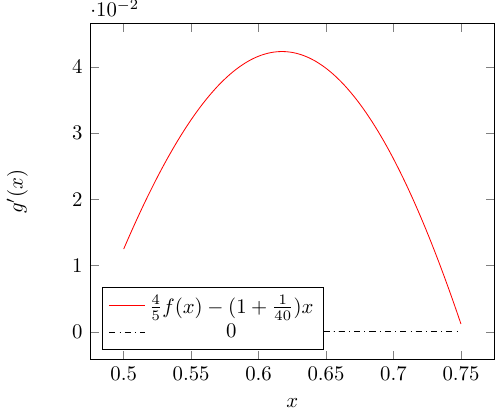}
    		\end{minipage}
	\caption{$g'(x)=-8x^6+ \frac{144}{5}x^5-36x^4+16x^3-\frac{41}{40}x$}
	\label{fig:g'(x)}
\end{figure}

\end{proof}

Another result we want to highlight shows that the $(6,3)$-median rule preserves agreement under network partitions because for any isolated group of at most $7n/10$ servers, their values will all be $\bot$ within $O(\log n)$ rounds, w.h.p. This will be very useful later, when we show how to recover from arbitrary blocking attacks. Specifically, we exploit that the adversary can only choose between a perpetual death spiral (if it keeps blocking a fixed set of $3n/10$ servers thus essentially halting commitment of commands entirely) or allowing committed commands to be spread among the useful servers, thus ensuring their survival.

\begin{lemma} \label{lem:permanently_blocking}
If at least a $\tfrac{3}{10}$-fraction of servers is permanently blocked, then irrespective of the number of useful servers, within $O(\log n)$ rounds no server is useful anymore, w.h.p.
\end{lemma}

\begin{proof}
Again, let $X_{t,i}$, $X_t$, $x_t$ and $f(x)$ be defined as in the proof of Lemma~\ref{lem:spiral}. Suppose that a $3/10$-fraction of the servers is permanently blocked. Given any $x_t \in [0,7/10]$, $\E[x_{t+1}] = (7/10) \cdot f(x_t)$. Set $g''(x) = (7/10) f(x) - (1-1/40)x$.
From Figure \ref{fig:g''(x)} it can be seen that
$g''(x) \le 0$ for all $x \in [0,1]$, which implies that $(7/10) f(x) \le (1-1/40)x$ for all $x \in [0,1]$. As the probability of a server to be successful is independent of the other servers, we can use Chernoff bounds to prove that for any $X_t \in [0,7n/10]$, $X_{t+1} \le (1-\varepsilon)X_t$ for some constant $\varepsilon>0$, w.h.p., which completes the proof.\qedhere

\begin{figure}
	\centering
    		\begin{minipage}[b]{0.5\linewidth}
   		 	\includegraphics[width=1\textwidth]{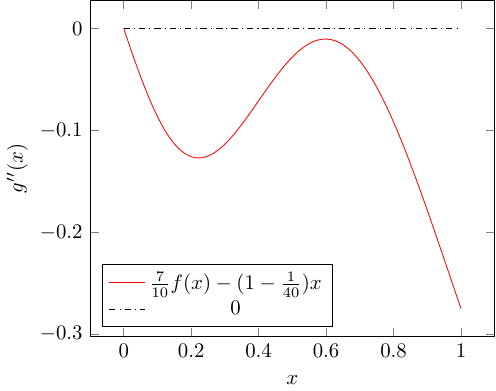}
    		\end{minipage}
	\caption{$g''(x)=-7x^6+ \frac{126}{5}x^5-\frac{63}{2}x^4+14x^3-\frac{39}{40}x$}
	\label{fig:g''(x)}
\end{figure}

\end{proof}

\subsection{Agreement}
We show that, perhaps interestingly, applying the $(6,3)$-median rule to the set of \textit{all servers} can be reduced to applying the $(3,3)$-median rule only to the set of \textit{useful servers}, which allows us to disregard useless severs and brings simplifications for the subsequent analysis of the agreement property.
Thus, it suffices to focus on the subsets of useful servers and to prove the following statement inspired by the analysis of the median rule in \cite{DoerrGMSS11}. More details and a generalization to a much wider class of ``$(k,\ell)$-rules'', as we will show next. The $(k,\ell,f)$-rule is defined in algorithm \ref{alg:(k,l,f)-rule}.

\begin{algorithm}
\caption{The $(k,l,f)$-rule}
\label{alg:(k,l,f)-rule}
\textbf{Preconditions:}
Consider any function $f:K^\ell \rightarrow K$ that is invariant to a reordering of its arguments and that satisfies validity, i.e., for all $x_1,\ldots,x_\ell \in K$, $f(x_1,\ldots,x_\ell) \in \{x_1,\ldots,x_\ell\}$.

\medskip
\textbf{Each server $i$ does the following each round:}
\begin{itemize}
\item send $k$ value requests to servers chosen uniformly and independently at random
\item if $x_i \not=\bot$ then for any value request received from some server $j$, send $x_i$ back to $j$
\item if at least $\ell$ replies are received, choose a subset of $\ell$ of these replies uniformly at random and set $x_i:=f(S)$, where $S$ is the multiset of values sent by these replies
\item if less than $\ell$ replies are received, set $x_i:=\bot$
\end{itemize}   
\end{algorithm}

Let $V$ be the set of all servers and $U_t$ be the set of useful servers in round $t$. The random experiment resulting from an execution of the $(k,\ell,f)$-rule can be modeled as $(R_1,R_2,R_3,\ldots)$, where $R_t =(R_{1,t},\ldots,R_{n,t})$ for all $t \ge 1$ and $R_{i,t} \in V^k$, $i\in \{1,\ldots,n\}$, denotes the $k$ servers chosen uniformly and independently at random by server $i$ in round $t$. 

From the definition of a 1-late adversary, we know that at the beginning of round $t$, the adversary is allowed to know $R_1,\ldots,R_{t-2}$ and the state of the system at the beginning of round $t-1$. Moreover, given $R_1,\ldots,R_{t-2}$ and the adversarial blocking decisions in rounds 1 to $t-1$, the set of useful servers $U_{t-1}$ and their values are uniquely defined. Thus, at the beginning of round $t$, a 1-late adversary knows $U_{t-1}$ and its values. However, since it does not know $R_{t-1}$, it does not know which of the servers that are non-blocked in round $t$ belong to $U_t$. Furthermore, since the adversary does not know $R_t$ as well, it does not know which of the non-blocked servers in round $t$ successfully apply the $(k,\ell,f)$-rule, nor does it know their outcome. Let $x(U_t)$ be the multiset of keys stored in $U_t$ and $S(i)$ be the event that server $i$ successfully applies the $(k,\ell,f)$-rule in round $t$. Furthermore, let $R(\ell,U_t)$ represent the random experiment in which every server in $V$ applies the $(\ell,\ell,f)$-rule to $U_t$. We show:

\begin{lemma} \label{lem:selection}
Let $i$ be any of the non-blocked servers in round $t$ and $x'_i$ be its key at the beginning of round $t+1$. For any key $x \in K$ and multiset of keys $X$ of size $|U_t|$ it holds: 
\[
  \text{Pr}_{R_t}[x'_i=x \mid S(i) \; \wedge \; x(U_t)=X] = \text{Pr}_{R(\ell,U_t)}[x'_i=x \mid x(U_t)=X]
\]
where $\text{Pr}_R$ represents the probability given random experiment $R$.   
\end{lemma}
\begin{proof}
First of all, note that due to the fact that the 1-late adversary does not know $R_t$ in round $t$, $\text{Pr}_{R_t}[x'_i=x]$ is independent of the fact that $i$ is non-blocked. Consider now the following alternative description of selecting $\ell$ requests in the $(k,\ell,f)$-rule:

Execute the $k$ requests \emph{one after the other}, where each request is sent to a server chosen uniformly and independently at random. Once $\ell$ replies have been received, stop sending further requests and apply $f$ to the received values.

Now, suppose that a server $i$ receives a set $M$ of at least $\ell$ replies when executing the $(k,\ell,f)$-rule, and let $L$ be any subset of these of size $\ell$. Certainly, a way to pick $L$ uniformly at random from $M$ is to pick a random permutation of the requests and take the first $\ell$ replies from that permutation. However, since the probability distribution for the value of each reply is independent of the other replies, taking the first $\ell$ replies of a fixed permutation results in the same probability distribution of the values as taking the first $\ell$ replies of a random permutation. Thus, $\Pr[x'_i=x \mid S(i)]$ is the same for the original $(k,\ell,f)$-rule and its alternative description above. The lemma now follows from the fact that given that the $i$th request reaches a server in $U_t$, the probability distribution of the reply is the same under the random experiments $R_t$ and $R(\ell,U_t)$ and independent of prior requests.
\end{proof}

Since the random choices of the servers in the random experiment $R_t$ are independent and also independent from the blocking decisions of a 1-late adversary in round $t+1$, Lemma~\ref{lem:selection} immediately yields the following corollary.

\begin{corollary} \label{cor:equivalency_median_rules}
For any round $t\ge 1$, any subset of servers $U_{t+1}$ that is non-blocked in rounds $t$ and $t+1$, and any multisets of keys $X,X'$,
\[
  \text{Pr}_{R_t}[x(U_{t+1})=X' \mid S(U_{t+1}) \; \wedge \; x(U_t)=X] = \text{Pr}_{R(\ell,U_t)}[x(U_{t+1})=X' \mid x(U_t)=X]
\]
where $S(U_{t+1})$ is the event that all servers in $U_{t+1}$ are successful in round $t$.
\end{corollary}

Hence, instead of studying the $(6,3)$-median rule, we will switch to the equivalent setting of applying the random experiment $R(3,U_t)$ to set $U_{t+1}$ in each round $t$, given that $U_{t+1}$ is the set of useful servers in round $t+1$. The only problem is that we do not know the useful sets beforehand since that depends on the blocking strategy of the adversary. As a way out, we will simply study the outcomes of the random experiments $R(3,U_t)$ for \emph{any} given sequence $(U_1,U_2,U_3,\ldots)$ of useful servers. In the next theorem, we only focus on sequences of useful servers with $|U_t| \ge n/4$ for all $t$ since we know from Lemma~\ref{lem:spiral} that for $|U_t| \le (1/3-\varepsilon)n$ a spiral-of-death would start, w.h.p., which would trivially satisfy agreement.

\begin{theorem}
\label{thm:stabilizing_consensus_3median}
For any sequence $(U_1, U_2, U_3, \dots)$ of sets of useful servers satisfying $n/4 \leq |U_t| \leq n$ for all $t$, the $(3,3)$-median rule applied by $U_{t+1}$ on $U_t$ for all $t\ge 1$ achieves agreement in $O(\log n)$ rounds, w.h.p.   
\end{theorem}

\begin{proof}
Let $n_t=|U_t|$. We start with the binary case, i.e., initially, $x_i \in \{0,1\}$ for all $i \in U_1$. Consider some fixed round $t$. Let $\gamma_t$ be the fraction of servers $i \in U_t$ with $x_i=0$ at round $t$. W.l.o.g. we assume that $\gamma_t \le 1/2$. We distinguish between three cases.

\subsubsection*{Case 1: $\gamma_t \le 1/4$}

Consider some fixed server $i \in U_{t+1}$ and let $N$ be the event that its value is 0. Since $\gamma_t \le 1/4$, 
\begin{align*}
  \Pr[N] & = {3 \choose 2} \gamma^2_t \cdot (1-\gamma_t) + \gamma^3_t \\
  & = 3 \gamma^2_t - 2 \gamma^3_t \le 3\gamma^2_t \le \frac{3}{4}\cdot \gamma_t
\end{align*}
Therefore, $\E[\gamma_{t+1}]\le (3/4)\gamma_t$. Since the event of having 0 in round $t+1$ is independent of the other servers in $U_{t+1}$ and $n_{t+1} \ge n/4$, the Chernoff bounds (Lemma~\ref{chernoff1}) imply that w.h.p., $\gamma_{t+1} \le \gamma_t \le 1/4$. Thus, we may assume that for polynomially many rounds $t' \ge t$, $\gamma_{t'}$ is upper bounded by $1/4$. In that case, it follows from the linearity of expectation that
\[
  \E[\gamma_{t+\tau}] \le (3/4)^{\tau} \gamma_t
\]
When choosing $\tau = (c+1) \log_{4/3} n$ for some constant $c>0$, it the Markov inequality implies that
\[
  \Pr[\gamma_{t+\tau}>0] = \Pr[\gamma_{t+\tau}\ge 1/n] \le n \cdot (3/4)^{\tau} = 1/n^c
\]
Thus, w.h.p., it takes at most $O(\log n)$ rounds until no server in $U_t$ stores a 0 and we have reached a consensus.

\subsubsection*{Case 2: $1/4 < \gamma_t \le 1/2-\sqrt{(c \ln n)/n}$}

Let $\delta_t = 1/2 - \gamma_t$ be the relative deviation from a balanced state. Then $\sqrt{(c \ln n)/n} \le \delta_t < 1/4$. We know from case 1 that
\begin{align*}
    \E[\gamma_{t+1}] & = 3\gamma_t^2 - 2\gamma_t^3 \\
    & = 3(1/2-\delta_t)^2 - 2(1/2-\delta_t)^3 \\
    & = 1/2 - (3/2)\delta_t + 2\delta_t^3
\end{align*}
Therefore,
\begin{align*}
   \E[\delta_{t+1}] & = (3/2)\delta_t - 2\delta^3_t \\
   & \geq (3/2)\delta_t - 2\cdot(1/4)^2\delta^3_t \\
   & \geq (11/8)\delta_t
\end{align*}
Since the event that a server has the value 0 in round $t+1$ is independent of the other servers, $n_{t+1} \ge n/4$, and $\gamma_t > 1/4$, it follows from the Chernoff bounds (Lemma~\ref{chernoff2}) that $\delta_{t+1} \ge (5/4)\delta_t$ w.h.p. This implies that within $\tau = \log_{5/4} n$ rounds, $\delta_{t+\tau}$ will be at least $1/4$, w.h.p., so we reached the situation of case 1. 
   
\subsubsection*{Case 3: $1/2-\sqrt{(c \ln n)/n} < \gamma_t \le 1/2$}

As in case 2, let $\delta_t = 1/2 - \gamma_t$, which means that now, $0 \le \delta_t < \sqrt{(c \ln n)/n}$. We prove the following lemma as a sub-claim of the overall proof. Specifically, we use the Central Limit Theorem to prove that with constant probability, we have a sufficiently large imbalance regardless of the previous imbalance.

\begin{lemma}
\label{lem:imbalance}
Let $\alpha>0$ be any constant. Then for any $\delta_t \ge 0$, $Pr[\delta_{t+1} \geq \alpha/\sqrt{n_{t+1}}] \geq \frac{1}{\sqrt{4\pi}(1+\sqrt{2}\alpha)} e^{-\alpha^2}$, provided $n_{t+1}$ is large enough.
\end{lemma}

\begin{proof}
For the proof of Lemma \ref{lem:imbalance}, we need the following notation. We say that a random variable $Y$ \textit{stochastically dominates} a random variable  $Z$, and write $Y \succeq Z$, if $Pr[Y \geq x] \geq Pr[Z \geq x]$ for any $x$. We start with a claim that follows from the proof of Claim 2.6 in \cite{DoerrGMSS11}.

\begin{claim} 
\label{claim:dominate}
For any two imbalances $\delta_t$ and $\delta_t'$ with $\delta_t \geq \delta_t' \geq 0$ it holds that $\delta_{t+1} \succeq \delta_{t+1}'$.
\end{claim}

Thus, we only need to prove Lemma~\ref{lem:imbalance} for $\delta_t=0$ since the general case follows from stochastic domination (see Claim~\ref{claim:dominate}).
For every available server $i \in \{1,\ldots,n_{t+1}\}$ at round $t+1$, let the random variable $X_i \in \{-1/2,1/2\}$ be defined as $X_i=1/2$ if and only if $x_i=0$. Note that in that case, $\delta_{t+1} = (1/n_{t+1}) \sum_{i=1}^{n_{t+1}} X_i$. For $\delta_t=0$ it holds that $\Pr[X_i=-1/2]=\Pr[X_i=1/2]=1/2$ and therefore, $\E[X_i]=0$ and $\V[X_i] = \E[X_i^2] - \E[X_i]^2 = 1/2$. Let $\Psi_{t+1} = \sum_{i=1}^{n_{t+1}} X_i$. Then it holds that $\E[\Psi_{t+1}]=0$. Furthermore, since the $X_i$'s are independent, $\V[\Psi_{t+1}] = n_{t+1}/2$. To proceed, we need the Central Limit Theorem, which is given formally in Lemma \ref{lem:central_limit}.

Thus, for any $\alpha>0$ (choose $a=\alpha$ and $b=\infty$),
\[
  \lim_{n_{t+1} \rightarrow \infty} \Pr[\Psi_{t+1} \ge \alpha \sqrt{n_{t+1}}] \ge 1-\Phi(\sqrt{2} \alpha) 
\]
which implies that
\[
  \Pr[\Psi_{t+1} \ge \alpha \sqrt{n_{t+1}}] \ge 1-\Phi(\sqrt{2} \alpha) - \varepsilon
\]
where $\varepsilon \to 0$ as $n_{t+1} \to \infty$. For $x \geq 0$, the value of $\Phi(x)$ can be bounded as follows (see \cite{DoerrGMSS11}):
\[
\frac{1}{\sqrt{2\pi}(1+x)} \cdot e^{-\frac{x^2}{2}} \leq 1-\Phi(x) \leq \frac{1}{\sqrt{\pi}(1+x)}  \cdot e^{-\frac{x^2}{2}}
\]
Therefore, we can lower bound the above probability by
\[
\frac{1}{\sqrt{2\pi}(1+\sqrt{2} \alpha)} \cdot e^{-\alpha^2} - \varepsilon \geq  \frac{1}{\sqrt{4\pi}(1+\sqrt{2} \alpha)} \cdot e^{-\alpha^2}
\]
Thus, if $n_{t+1}$ is large enough then for any $\alpha>0$, 
\[
  \Pr[\delta_{t+1} \ge \alpha/\sqrt{n_{t+1}}] \ge \frac{1}{\sqrt{4\pi}(1+\sqrt{2}\alpha)} e^{-\alpha^2} .
\]
\end{proof}

Moreover, we can show the following lemma.

\begin{lemma}
For any $\delta_t>0$,
  $\Pr[ \delta_{t+1} \le \frac{4}{3} \delta_t] \le 2e^{-\Theta(\delta_{t}^2 n_{t+1})}$.
\end{lemma}
\begin{proof}
Recall that $\Psi_{t+1}= \delta_{t+1} \cdot n_{t+1}$. Thus,
\begin{align*} 
  \Pr[ \delta_{t+1} \le \frac{4}{3} \delta_t] 
  & = \Pr[ \Psi_{t+1} \leq \frac{4}{3} \cdot \frac{n_{t+1}}{n_t} \cdot  \Psi_{t}]  
\end{align*}
Analogously to Case 2, it holds that $\E[\delta_{t+1}] \ge \frac{11}{8} \delta_t$, and therefore,
\begin{align*}
E[ \Psi_{t+1}] 
&\geq \frac{11}{8} \cdot \frac{n_{t+1}}{n_t} \cdot \Psi_t  
\end{align*}
With $\eta = 1/33$ we get:
\begin{align*}
   \Pr[\delta_{t+1} \leq \frac{4}{3} \delta_t] 
   &= \Pr[\delta_{t+1} \leq \frac{32}{33} \cdot \frac{11}{8} \delta_t]\\
   &= \Pr[\delta_{t+1} \leq (1-\frac{1}{33}) \cdot \frac{11}{8} \delta_t]\\
   &\le \Pr[\delta_{t+1} \leq (1-\eta) \cdot \E[\delta_{t+1}]]\\
\end{align*}
Using the Chernoff bounds (Lemma \ref{chernoff2}), it follows that:
\begin{align*}
  \Pr[\Psi_{t+1} \leq (1-\eta) \cdot \E[\Psi_{t+1}]]
  & \le \Pr[|\Psi_{t+1} - \E[\Psi_{t+1}]| \geq \eta \cdot \E[\Psi_{t+1}]] \\
  & \leq 2e^{-2 \cdot(\eta \cdot \E[\Psi_{t+1}])^2/ \sum_{i=1}^{n_{t+1}}(-\frac{1}{2}-\frac{1}{2})^2}\\
  & \leq e^{\frac{-2 \cdot \eta^2 \cdot (\frac{11}{8} \cdot \frac{n_{t+1}}{n_t} \cdot \Psi_t)^2}{n_{t+1}}}\\
  & \leq e^{\frac{- \frac{121}{32} \cdot \eta^2 \cdot (n_{t+1} \cdot \delta_t)^2}{n_{t+1}}}\\
  & = e^{- \Theta(\delta_t^2n_{t+1})}.\hfill\qedhere
\end{align*}
\end{proof}

Thus, there is an increasingly stronger drift to $\sqrt{(c\ln n)/n}$ as $\delta_t$ gets larger so that similarly to Lemma 2.8 in \cite{DoerrGMSS11}, the following result can be shown:

\begin{lemma}
If at a round $t$ we have $\delta_t < \sqrt{(c \ln n)/n}$ then there is a round $t'=t+O(\log n)$ with $\delta_{t'} \ge \sqrt{(c \ln n)/n}$, w.h.p.
\end{lemma}

The intuition behind the lemma is that when starting with $\delta_t \ge \alpha/\sqrt{n_t}$ for a sufficiently large constant $\alpha>0$, with constant probability there is a so-called \emph{run} of consecutive increases of $\delta_t$ by a factor of at least $4/3$ till $\delta_t \ge \sqrt{(c \ln n)/n}$. Furthermore, with constant probability, $\delta_{t+1} \ge \alpha/\sqrt{n_{t+1}}$ whenever $\delta_t < \alpha/\sqrt{n_t}$ (see Lemma~\ref{lem:imbalance}), and if afterwards the run fails to reach $\delta_t \ge \sqrt{(c \ln n)/n}$ (due to some increase of $\delta_t$ of less than $4/3$), this happens on expectation after a constant number of increases by at least $4/3$. Thus, w.h.p., at most $O(\log n)$ runs are needed till a round $t$ is reached with $\delta_t \ge \sqrt{(c \ln n)/n}$, and the time covered by these runs is $O(\log n)$, w.h.p.

\paragraph{\textbf{The case of arbitrary values}}
Suppose that the servers $U_t$ are numbered from 1 to $n_t$, and let us define the \emph{identifier} $id(x_i)$ of a value $x_i$ at server $i$ as $id(x_i)=x_i \circ i$ (i.e., $x_i$ is padded with lower order bits representing $i$). These id's can be uniquely ordered so that we may assume, w.l.o.g., that $id(x_1)<\ldots<id(x_{n_t})$. We associate with each $id(x_i)$ a value $g(i)$ called \emph{gravity}, which is the probability that a server chooses $x_i$ as its median in the next round. The gravity can be computed as follows: 
value $x_i$ may either be chosen at least twice, or it is the median of the three chosen values and the values are pairwise distinct. (Note that the median rule does not make use of the server numbering, but it certainly holds that if $x_i$ is the median according to our id ranking, then a value equal to $x_i$ would also be the outcome of the original median rule.) Thus,
\begin{align*}
   g(i) & = {3 \choose 1} \frac{1}{n_t} {2 \choose 1} \frac{i-1}{n_t} \cdot \frac{n_t-i}{n_t} + {3 \choose 2} \left( \frac{1}{n_t} \right)^2 \cdot \frac{n_t-1}{n_t} + \left( \frac{1}{n_t} \right)^3 \nonumber \\
    & = \frac{6(i-1)(n_t-i) +3n_t -2}{n_t^3} 
\end{align*}
As can be easily checked, this is maximized for $i= \lceil n_t/2 \rceil$, with $g(i)$ being roughly $\frac{3}{2 \cdot n_t}$ for such an $i$. For each value $x$ and round $t$, let us define the set of \emph{heavy} identifiers $\mathcal{H}_{x,t}$ as the subset of the $\Phi_t = 2\sqrt{c n_t \ln n_t}$ identifiers $id(x_i)$ with $x_i=x$ of largest gravity (or all such values if there are less than $\Phi_t$), where $c>0$ is a sufficiently large constant. By definition, $0 \le |\mathcal{H}_{x,t}|\le \Phi_t$. Following along the same line of arguments as in \cite{DoerrGMSS11}, the following lemma can be shown:

\begin{lemma} \label{lem:reduction}
For any initial set of values, it takes at most $O(\log n)$ rounds until at least one of the following two cases holds for all values $x$ w.h.p.:
\begin{enumerate}
\item At least one identifier $id(x_i) \in \mathcal{H}_{x,t}$ satisfies $g(i) < (4/3)/n_t$ or $|\mathcal{H}_{x,t}|=0$, or
\item $|\mathcal{H}_{x,t}|=\Phi_t$
\end{enumerate}
\end{lemma}
\begin{proof}
The intuition behind that lemma is the following. If Case 1 does not apply for some value $x$ then $\mathcal{H}_{x,t}$ only contains id's of gravity at least $(4/3)/n_t$. In that case, $\E[|\{ id(x_i)$ in round $t+1$ with $x_i=x\}|] \ge (n_{t+1}/n_t) \cdot (4/3)|\mathcal{H}_{x,t}|$. Thus, if we already have $|\mathcal{H}_{x,t}|=\Phi_t$ then it follows from the Chernoff bounds that $|\mathcal{H}_{x,t+1}|=\Phi_{t+1}$ w.h.p. Otherwise, as long as Case 1 does not apply, $\E[|\{ id(x_i)$ in round $t$ with $x_i=x\}|]$ increases by a constant factor on expectation in each round $t$, which will reach $\Phi_t$ within $O(\log n)$ rounds, w.h.p. To prevent Case 2 for some $x$, we therefore have a reach a round $t$ where Case 1 applies. Once Case 1 applies, we distinguish between two cases. If $|\mathcal{H}_{x,t}|=0$ then, certainly, also $|\mathcal{H}_{x,t+1}|=0$. Otherwise, if at least one identifier $id(x_i) \in \mathcal{H}_{x,t}$ satisfies $g(i) < (4/3)/n_t$ then we distinguish between two further cases. If $x$ is larger than the median value then consider the values to be cut into two groups $G_0=\{ x_i \mid x_i < x\}$ and $G_1=\{ x_i \mid x_i \ge x\}$. From the definition of $\mathcal{H}_{x,t}$ it follows that $G_0-G_1 \ge \sqrt{c n \ln n}$, i.e., we can view the system as being in case 2 of the binary case, where $G_0$ represents majority value and $G_1$ represents the minority value. Since the majority value continuously grows, w.h.p., $\mathcal{H}_{x,t}$ will continue to have a value of gravity less than $(4/3)/n_t$, w.h.p. The case for $x$ being smaller than the median value is similar, implying Theorem~\ref{thm:stabilizing_consensus_3median}.
\end{proof}
This concludes the proof of Theorem \ref{thm:stabilizing_consensus_3median}.
\end{proof}

Combining Lemma~\ref{lem:spiral} (if the fraction of useful servers drops below $1/3$, it will rapidly spiral to 0) with Theorem~\ref{thm:stabilizing_consensus_3median} (if $n/4$ servers are continuously useful we reach agreement quickly), we obtain agreement in the general setting. Note that a blocking bound on the adversary is not even needed since the $(6,3)$-median rule trivially achieves agreement on $\bot$ if the fraction of useful servers ever drops below $1/3$.

\begin{corollary}
\label{cor:stabilizing_consensus_quarter_available}
For any 1-late adversary, the $(6,3)$-median rule achieves agreement in $O(\log n)$ rounds, w.h.p.
\end{corollary}

When combining this corollary with Lemma~\ref{lem:available} ($3n/4$ servers are useful under a limited adversary), we obtain the main result of this section.

\begin{theorem} 
\label{th:keys}
If the adversary is 1-late and at least $n/4$ servers are useful in each round, the $(6,3)$-median rule solves the stabilizing consensus problem (Definition~\ref{def:stabilizing_consensus}) in $\bigO(\log n)$ rounds, w.h.p. Furthermore, given a 1-late $1/10$-blocking adversary, at least $3/4$ of the servers are useful w.h.p., for any number of rounds that is polynomial in $n$.
\end{theorem}

\section{SMR with an Extended Median Rule} \label{sec:xmedian}

We extend our consensus protocol to a solution of the stabilizing SMR problem under a blocking adversary (Definition \ref{def:stabilizing_smr}). We call a command \emph{injected} once it is first received from the client by a non-blocked server. Here, clients are expected to send the same command multiple times to a random server to improve the chance of a successful injection (as a command sent to a blocked server is lost). We assume here w.l.o.g.\ that every client command is unique (a mechanism to resolve ambiguities is introduced in the next section.) The details of the extended $(k,\ell)$-median rule can be found in Algorithm \ref{alg:extended}.

\begin{algorithm}
\caption{The extended $(k,\ell)$-median rule}
\label{alg:extended}
\textbf{Preconditions:}
Each server $i$ maintains a log $L_i \in K^*$, where $K^*$ contains all finite sequences of commands from a totally ordered set $K$ without repetitions. We use lexicographic order to impose a total order on these logs. Initially, $L_i = (x_0)$ for all $i$, where $x_0$ is a seed command. Every round, any set of client commands $x \in K$ can be injected at any subset of servers. 

\smallskip

\textbf{Each server $i$ does the following each round:}
\begin{itemize}
\item for every command $x$ received from a client that is not yet contained in $L_i$, server $i$ sends $\sigma \log n$ \emph{append requests} to servers chosen uniformly and independently at random, for a sufficiently large constant $\sigma \ge 1$
\item server $i$ sends $k$ \emph{log requests} to servers chosen uniformly and independently at random
\item if $L_i \!\neq\! \bot$ then for any \emph{log request} received by server $i$ from a server $j$, $i$ sends $L_i$ back to $j$ 
\item if server $i$ receives at least $\ell$ replies
\begin{itemize}
    \item server $i$ chooses a subset $M$ of size $\ell$ from the received logs uniformly at random
    \item server $i$ sets $L_i:=L'_i \circ \bar{L}$, where $L'_i$ is the median of $M$ and $\bar{L}$ contains (in any order) all values in logs in $M$ or in \emph{append requests} received in that round that are \emph{not} in $L'_i$
\end{itemize}
\item if server $i$ receives less than $\ell$ logs, it sets $L_i:=\bot$
\end{itemize}
\end{algorithm}

The first step amplifies the presence of any newly injected command $x$, giving it a high probability to reach all servers in subsequent rounds. Afterwards, the servers request logs from each other and select the median log from the responses, appending any missing commands. This process ensures that, over time, the servers converge on an ordering of commands that have circulated long enough. Similarly to the previous section, we call a server \emph{useful} in round $t$ if it stores a non-empty log at the beginning of $t$ and is non-blocked during round $t$. We show the following theorem. 

\begin{theorem} \label{th:sequence}
If the adversary is 1-late and at least $n/4$ servers are useful in each round, the extended $(6,3)$-median rule solves the stabilizing SMR problem (Definition \ref{def:stabilizing_smr}) in $\bigO(\log n)$ rounds, w.h.p. Furthermore, given a 1-late $1/10$-blocking adversary, at least $3/4$ of the servers are useful for polynomially many rounds, w.h.p.
\end{theorem}

The intuition behind the proof of Theorem \ref{th:sequence} is as follows: 
\begin{itemize}
  \item \textbf{Reaching a critical mass.} Once a client request with value $x$ has been received by some server~$i$, w.h.p.\ $\Theta(\log n)$ useful servers append $x$ to their logs in the next round.

  \item \textbf{Broadcast.} After $T_B=O(\log n)$ rounds, $x$ is part of every non-empty log $L_i$, w.h.p.

  \item \textbf{Stabilization.} Once a command occurs in the prefix of all non-empty logs, the analysis can be reduced to the median rule on single values to show that after $T_M = O(\log n)$ rounds all useful servers place $x$ at the same index, w.h.p.

\end{itemize}

\subsection{Broadcasting with the $(k,\ell)$-Gossip Rule}

We begin by examining how a command injected at an unblocked server spreads. The main insight is that once $\bigO(\log n)$ servers hold a value $x$, it disseminates to all servers exponentially fast, even under a blocking adversary. In particular, this mechanism can be used to broadcast information to unblocked servers. We formalize this by introducing the $(k,\ell)$-gossip rule, which captures how commands propagate through the logs under the extended $(k,\ell)$-median rule in a general way. Defining the $(k,\ell)$-gossip rule also allows us to reuse these results in later sections.

\begin{algorithm}
\caption{$(k,\ell)$-Gossip Rule}
    \textbf{Variables:}
    
    Each server $i$ has a local variable $x_i$. There is a subset of servers that start with $x_i = x$, all others with $x_i = x_{0}$, where $x_0$ is a dummy value. 

    \medskip

\textbf{Each server $i$ does the following each round:}
    \begin{itemize}
        \item send $k$ value requests to servers chosen uniformly and independently at random 
        \item if $x_i \not=\bot$ then for any value request received from some server $j$, send $x_i$ back to $j$ 
        \item if at least one reply was received that contains $x$, set $x_i := x$, else set $x_i := x_0$
        \item if less than $\ell$ values are received, set $x_i:=\bot$
    \end{itemize}
\end{algorithm}

If we consider only a single command $x$ that starts in some server logs, then the spreading behavior of $x$ through the server logs in the extended $(6,3)$-median rule is equivalent to the spreading of $x$ in the $(6,3)$-gossip rule. We show the following theorem. 

\begin{theorem}
    \label{thm:gossip_rule}
If the adversary is 1-late, at least $n/4$ servers are useful in each round, and at least $c \log n$ useful servers initially have $x$ for some sufficiently large constant $c \ge 1$ then the $(6,3)$-gossip rule ensures that within $T_B=\bigO(\log n)$ rounds, every useful server has $x$ in its log, w.h.p. If less than $c \log n$ useful servers initially have $x$, then within $T_B$ rounds either all servers have $x$ or none has $x$.
\end{theorem}

We formally prove this theorem over a series of lemmas. First we show that when we start out from a round where some $\Omega(\log n)$ servers have a value $x$, then the number of servers that have the command in their logs will grow by a factor bigger than 1 each round, until a certain fraction of servers have $x$.

\begin{lemma}    
    \label{lem:increase_servers_with_command}
    Assume the adversary is 1-late and there are at least $n/4$ useful servers in round $t$ and $t+1$. Further, assume that at least $c \log n$ servers but at most a fraction of $\gamma_t \leq 1/3$ of the useful servers have $x$ in round $t$ (for constant $c >0$). Then a fraction of at least $\tfrac{3}{2}\cdot \gamma_t$ of the useful servers have $x$ in round $t+1$, w.h.p.
\end{lemma}

\begin{proof}
    Assume server $i$ is useful in round $t+1$. Let $U_t$ be the subset of useful servers in round $t$ with $x \in L_i$ for all $i \in U_t$. Further, let $n_t := |U_t|$.
    
    By Corollary \ref{cor:equivalency_median_rules} we know that for server $i$ applying the $(k,\ell)$-gossip rule on all servers is equivalent to applying the $(\ell, \ell)$-gossip rule on the set of useful servers.

    Let $p_t$ be the probability that useful server $i$ requests the log from at least one of the servers in $U_t$ is at least 
    \begin{align*}
        p_t & = 1 - \smash{\tfrac{(n_t-\gamma_t n_t)^3}{n_t^3}} \\
        & = 1 - \big( 1 - 3\gamma_t + 3\gamma_t^2 - \gamma_t^3\big)\\
        & = 3\gamma_t - 3\gamma_t^2 + \gamma_t^3
        \geq 3\gamma_t - 3\gamma_t^2 \\ 
        & \geq 3\gamma_t - 3\gamma_t\cdot \tfrac{1}{3} \tag*{\text{$\gamma_t \leq 1/3$}}
        \geq 2\gamma_t
    \end{align*}

    For some server $i \in U_{t+1}$, let $X_{it} = 1$ if server $i$ has $x$ after round $t+1$ else $X_{it} = 0$. We have $\Pr(X_{it} = 1) \geq p_t$. Then the number of useful servers in round $t+1$ that will have $x$ in their logs is $X_t := \sum_{i \in U_{t+1}} X_{it}$, which is a sum of independent Bernoulli variables. The expected value is at least \smash{$\E[X_t] = \sum_{i \in U_{t+1}} \, p_t \geq 2 \cdot \gamma_t \cdot n_{t+1} 
    =: \mu_t$}. Further, it is $\tfrac{n}{4} \leq n_t, n_{t+1} \leq n$ by the precondition of this lemma. It follows that $\mu_t \geq  c \log n \cdot \tfrac{n_t}{n_{t+1}} \geq \tfrac{nc}{4n} \log n = \tfrac{c}{4} \log n$.  
    
    With a Chernoff bound (see Lemma \ref{chernoff1}) we obtain
    \begin{align*}
        \Pr\big(\gamma_{t+1} \leq \tfrac{3\gamma_t}{2}\big) & = \Pr\big(n_{t+1}\gamma_{t+1} \leq \tfrac{3 \gamma_t n_{t+1}}{2}\big) \\
        & = \Pr\big(X_t \leq \tfrac{3 \gamma_t n_{t+1}}{2}\big) \\
        & = \Pr\big(X_t \leq (1-\tfrac{1}{4}) \mu_t\big) \\
        & \leq \exp\big( \!-\!\tfrac{\mu_t}{32} \big) \leq \exp\big( \!-\!\tfrac{c \log n}{128} \big) = n^{-c/128}.
    \end{align*}
    This means that w.h.p.\ $\gamma_{t+1} > \tfrac{3\gamma_t}{2}$.
\end{proof}

The previous lemma implies that the fraction of useful servers that have $x$ grows exponentially, until at least a third of all useful servers have it. Next we show that once we are above the 1/3-fraction of useful servers that have $x$, we will not drop below it w.h.p.

\begin{lemma}
    \label{lem:servers_with_command_stable}
    Assume the adversary is 1-late and at least $\tfrac{n}{4}$ servers are useful in round $t$ and $t+1$. If a fraction of $\gamma_t > 1/3$ useful servers have $x$ in round $t$, then $\gamma_{t+1} > 1/3$, w.h.p.
\end{lemma}

\begin{proof}
    Reconsider the probability $p_t$ that a server server $i$ that is useful in round $t+1$ will have $x$ after round $t$ (discussed in the proof of the previous lemma). We have $p_t = 1 - (1 - \gamma_t)^3 > 1 - (1 - \tfrac{1}{3})^3 = \tfrac{19}{27} > \tfrac{2}{3}$.
    If $X_t$ is the random number of useful servers that have $x$ in round $t+1$, we have $\E[X_t] \geq 2n_{t+1}/3 =: \mu_t$. We can then use the Chernoff bound from Lemma \ref{chernoff1} and the fact that $\mu_t \geq n/6$ (due to $n_{t+1} \geq n/4$), to bound the tail probability that less than half the expected number of useful servers have $x$.
    \[
        \Pr\big(\gamma_{t+1} \leq \tfrac{1}{3}\big) = \Pr\big(X_t \leq \tfrac{n_{t+1}}{3} \big) = \Pr\big(X_t \leq \tfrac{\mu_t}{2}\big) \leq e^{-\mu_t/8} \leq e^{-n/48},
    \]
    thus the claim holds w.h.p.
\end{proof}

The next step is to show that, once a third of the useful servers already have $x$, the proportion of servers still lacking $x$ decreases by a constant factor in each subsequent round, on expectation.

\begin{lemma}\label{lem:decrease_servers_without_command}
    Assume the adversary is 1-late and that at least $n/4$ servers are useful in the current round $t$ and next round $t+1$. Further, a fraction of $\gamma_t > 1/3$ useful servers have $x$ in round $t$. Let $\delta_t = 1 - \gamma_t$ be the fraction of useful servers that do not have $x$ in their log in round $t$. Then $\E[\delta_{t+1}] < \tfrac{1}{2}\delta_t$.
\end{lemma}

\begin{proof}    
    By Corollary \ref{cor:equivalency_median_rules}, using the $(k,\ell)$-gossip rule is equivalent to a useful server in round $t+1$ sampling three servers from the set of useful servers in round $t$. The probability $p_t$ of some useful server to sample no useful server that has $x$ is at most
    \begin{align*}
        p_t = \delta_t^3 & = (1-\gamma_t)^2 \delta_t \\ 
        & \hspace{-3.2mm}\stackrel{\gamma_t > 1/3}{<} \tfrac{4}{9} \cdot \delta_t < \tfrac{1}{2} \cdot \delta_t.
    \end{align*}
    The claim follows from the fact  $\E[\delta_{t+1}] = p_t$.
\end{proof}

We can now combine the two previous lemmas to prove that given that $1/3$ of useful servers already have $x$, all will have $x$ after $\bigO(\log n)$ rounds.

\begin{lemma}
    \label{lem:servers_without_command_zero}
    Assume that the adversary is 1-late and at least $n/4$ servers are useful in rounds $t$ and the next $\tau$ rounds $t+1, \dots, j+\tau$. Further assume that a fraction $\gamma_t > 1/3$ of useful servers have $x$ in round $t$. Then $\gamma_{t+\tau}=1$, w.h.p., for some suitable $\tau \in O(\log n)$.
\end{lemma}

\begin{proof}
    Due to Lemma \ref{lem:servers_with_command_stable} and using the union bound from Lemma \ref{lem:unionbound_whp}, the condition $\gamma_{t+s} > 1/3$ is maintained w.h.p., in each round $t+s$, as long as $s$ is polynomially bounded in $n$. Let $\delta_t := 1- \gamma_t$.
    Since $\gamma_{t+s} > 1/3$ we can apply Lemma \ref{lem:decrease_servers_without_command} and we have that $\E[\delta_{t+s}] \leq \tfrac{1}{2} \cdot \delta_{t+s-1}$ for any polynomially bounded $s$.
    This implies $\E[\delta_{t+\tau}] \leq \big(\tfrac{1}{2}\big)^\tau  \delta_{t}$ by linearity of expectation.
    We choose $\tau = (c+1) \log_2 n$ for some constant $c>0$, then by Markov's inequality we obtain
    \[
      \Pr[\delta_{t+\tau}>0] = \Pr[\delta_{t+\tau}\ge 1/n] \le n \cdot (1/2)^{\tau} = 1/n^c.
    \]
    Consequently, we have that $\delta_{t+\tau} = 0$, w.h.p.
\end{proof}

We have the required pieces for the proof of Theorem \ref{thm:gossip_rule}.

\begin{proof}[Proof of Theorem \ref{thm:gossip_rule}]
    In the following we consider a polynomially bounded number of events each of which holds w.h.p. The union bound from Lemma \ref{lem:unionbound_whp} then implies that \textit{all} of these events hold w.h.p.
    
    If less than $c \log n$ useful servers have $x$, then we model the number of servers that have $x$ as a random walk $X_t \in \{0, \ldots ,q\}$ (with $q = \lceil c \log n \rceil$), where both the states 0 and $q$ are absorbing. We know from the proof of Lemma \ref{lem:increase_servers_with_command} that this random walk becomes more biased towards $q$ as we move closer to $q$, i.e., we have the probability  $\Pr(X_{t+1} \geq \tfrac{3}{2}X_t) \geq 1-e^{-c'X_t}$ for some $c' >0$. This lets us apply a claim from \cite{DoerrGMSS11} reproduced in Appendix \ref{sec:probabilistic_concepts} Lemma \ref{lem:dichotomy}: after $\bigO(\log n)$ rounds the random walk will be in state 0 or $q$.    
    
    If $c \log n$ useful servers have $x$, the requirement of Lemma \ref{lem:increase_servers_with_command} is satisfied, which means that the fraction $\gamma_t$ of useful servers with $x$ increases by a factor \smash{$\gamma_{t+1} = \tfrac{3\gamma_t}{2}$} with each round as long as $\gamma_t \leq \tfrac{1}{3}$, w.h.p.  
    This implies that after $\bigO(\log_{3/2} n)$ rounds, \smash{$\gamma_t > \tfrac{1}{3}$} w.h.p. We learn from Lemma \ref{lem:servers_without_command_zero} that \smash{$\gamma_t = 1$} after $\bigO(\log n)$ further rounds w.h.p.
\end{proof}

As preparation for a later section, where we introduce a protocol that recovers from an adversary capable of arbitrarily blocking servers, we present the $(6,3)$-priority rule. In essence, this rule instructs servers to select the highest-priority value among at least $\ell$ sampled values.

\begin{definition}[$(k,\ell)$-Priority Rule]
    \label{def:priority_rule}
    Each server $i$ starts with a value stored locally in variable $x_i$ where the values are in total order. In each round, each server does the following:
    \begin{itemize}
        \item send $k$ value requests to servers chosen uniformly and independently at random 
        \item if $x_i \not=\bot$ then for any value request received from some server $j$, send $x_i$ back to $j$ 
        \item if at least $\ell$ values are received, choose a subset of $\ell$ replies and set $x_i$ to the largest of the values in the replies
        \item if less than $\ell$ values are received, set $x_i:=\bot$
    \end{itemize}
\end{definition}

By Theorem \ref{thm:gossip_rule}, the $(6,3)$-priority rule ensures agreement on a single value given that only constantly many are in circulation and any maximum value initially held by $\Omega(\log n)$ useful servers will rapidly propagate to all useful servers in $\bigO(\log n)$ rounds.

\begin{corollary}
    \label{cor:broadcast_priority_rule}
If the adversary is 1-late, at least $n/4$ servers are useful in each round, and least $c \log n$ useful servers initially have the value $x$ with highest priority, for some sufficiently large constant $c \ge 1$, then the $(6,3)$-priority rule ensures that within $T_B=\bigO(\log n)$ rounds, every useful server has $x$, w.h.p. If the number of different priority values in circulation is constant, then the $(6,3)$-priority rule ensures that all servers hold a single after $T_B$ rounds, w.h.p.
\end{corollary}

\subsection{Solving the Stabilizing SMR Problem}

Recall that whenever a client command $x$ is injected at a non-blocked server, it will be sent to $\sigma \log n$ servers chosen uniformly and independently at random. Given that $n/4$ servers are useful in each round, it easily follows from the Chernoff bounds that at least $c \log n$ useful servers will receive $x$ in the next round, w.h.p., given that $c=\sigma/5$ and $\sigma \ge 1$ is a sufficiently large constant. Thus, Theorem~\ref{thm:gossip_rule} implies the following result.

\begin{corollary}
    \label{cor:broadcast_command_median}
    Assume that the adversary is 1-late and a command $x$ is injected at a non-blocked server. Then the extended $(6,3)$-median rule ensures that $x$ is in the log of all useful servers after $T_B = \bigO(\log n)$ rounds, w.h.p., given that at least $n/4$ servers are useful for $T_B$ rounds.
\end{corollary}

The next lemma assumes that broadcasting $x$ to all useful servers has been completed. It is based on an observation that is central to implementing a state machine protocol with the median rule: namely, the median rule on complete logs can be reduced to applying the median rule to prefixes of logs.

\begin{lemma}
    \label{lem:new_command_stable}
    Suppose a command $x$ appears in the log of every useful server. Then, w.h.p., after $T_M = O(\log n)$ rounds, the shortest prefixes of logs containing $x$ of all servers with non-empty logs are all equal, provided there are at least $n/4$ useful servers in each of those $T_M$ rounds.
\end{lemma}

\begin{proof}
For each round $t$, let $U_t$ be the set of all useful servers, so $\tfrac{n}{4} \le |U_t| \le n$. For $i \in U_t$, let $L_i$ be its log and let $\pi_{ij} \le L_i$ be the shortest prefix containing $x$ (where $\le$ denotes the prefix relation). Define $p_t := \{\pi_{ij}\mid i \in U_t\}$. We show that $P_{t+1}\subseteq p_t$. 

According to the $(k,\ell)$-median rule, each server’s log $L_i$ can evolve in one of two ways: (1) If server $i$ receives fewer than $\ell$ responses (or is blocked), it sets $L_i:=\bot$, thus ceases to be useful. (2) Otherwise, server $i$ sets $L_i:=L'_i\circ \bar{L}$, where $L'_i$ is the median among sampled logs and $\bar{L}$ consists of all commands from those logs (including server $i$’s own) that are not already in $L'_i$. Since all sampled logs already contain $x$, the resulting prefix $\pi_{i,j+1}$ remains in $p_t$. Thus, $P_{t+1}\subseteq p_t$. 

Moreover, because we rely on lexicographic ordering, any commands after $x$ do not affect which prefix is chosen by the median rule. Hence servers effectively apply the $(k,\ell)$-median rule on the set of values from $p_t$ and therefore converge on a single prefix in $p_t$ by Theorem~\ref{th:keys}. 
\end{proof}

It remains to fit the pieces together to prove Theorem \ref{th:sequence}.

\begin{proof}[Proof of Theorem \ref{th:sequence}]
    \phantom{a}
    \begin{itemize}        
        \item {\bf Agreement:} By Corollary \ref{cor:broadcast_command_median} every useful server has $x$ after $T_B = \bigO(\log n)$ rounds. This satisfies the precondition of Lemma \ref{lem:new_command_stable}, which shows that after and additional $T_M = \bigO(\log n)$ rounds there is agreement on the shortest prefixes containing $x$ among servers with non-empty logs. This implies that the position of $x$ is fixed in all logs after a total of $T_E= T_B + T_M = \bigO(\log n)$ rounds.
        \item {\bf Availability}: Since $n/4$ of the servers are continuously useful, this naturally satisfies the availability condition. 
        \item {\bf Validity}: During the execution of the algorithm, a server only adds commands to its log that were either contained of other logs or were injected by some client and subsequently spread with append requests.
    \end{itemize}
The statement about the $1/10$-blocking adversary directly follows from Lemma \ref{lem:available}.
\end{proof}

\section{The Compact $(k,\ell)$-Median Rule} \label{sec:cmedian}

As logs grow, storing and transmitting entire logs becomes increasingly inefficient. Our aim is to solve the stabilizing SMR problem so that, once a command is committed, servers only need to maintain the shared state and the non-committed portion of the log. By Theorem~\ref{th:sequence}, any injected command appears at the same position in every non-empty log within at most $\tau \log n$ rounds (for a sufficiently large constant $\tau \ge 1$), with high probability. Therefore, each server needs to keep a command in its log for just $T \geq \tau \log n$ rounds; after that, the command is safe to consider committed and can be executed on the shared state.

However, some subtleties arise as achieving a consensus on the age of a command is not straightforward. Simply having the first server to receive a client request $x$ attach a timestamp does not work since client may inject a command $x$ multiple times, potentially resulting in different timestamps. This is stipulated by our protocol for submitting commands on the client side: repeatedly send the command to some server (best picked at random each time) until commitment is acknowledged by the contacted server. Resubmitting commands is in fact necessary given that client-server communication is asynchronous and servers can be blocked. Because servers no longer store the entire log, there is the danger to execute the same command multiple times if two client requests with $x$ arrive more than $T$ rounds apart.

To address these issues, the servers require each client to include a sequence number (along with its client ID) in every command. Each server maintains, for each client $c$, a variable $\textit{sn}(c)$ representing the most recently committed sequence number of $c$. Initially, $\textit{sn}(c) = 0$. A server $v$ will accept a command $x$ from $c$ with sequence number $\textit{sn}(x)$ only if $\textit{sn}(x) = \textit{sn}(c) + 1$. If this condition is met, $v$ spreads $x$ to $\sigma \log n$ other servers (for a sufficiently large constant $\sigma \ge 1$), along with the current round number to track the command’s age. Because servers may disagree on $x$’s age, the median rule is employed to settle on a single value. Once a server commits and executes $x$, it sets $\textit{sn}(c) = \textit{sn}(x)$ and stores the position of $x$ in its local sequence of committed commands in $\textit{ps}(c)$.

A further concern arises if a client submits multiple commands for the same sequence number, which could lead to ambiguity about which command to commit. To prevent this, we choose $T > 2T_B$, where $T_B$ is the time to broadcast a command. If a server sees two distinct commands for the same sequence number, it replaces both with $\bot$ (a null command). The condition $T > 2T_B$ ensures that from the moment a server first receives a command $x$ for sequence number $s$, at most $T_B$ rounds pass before $x$ reaches all other servers; consequently, any second command $x'$ for $s$ must arrive at the original server within those same $T_B$ rounds for it to be accepted. After an additional $T_B$ rounds, $x'$ would propagate and prompt all servers to switch to $\bot$ for sequence number $s$ before it is committed. However, for simplicity, we assume in the \textit{pseudo-code} that no client sends different commands for the same sequence number. 
Algorithm \ref{alg:compact} covers the details of the compact median rule. 

\begin{algorithm}
\setstretch{0.98}
\caption{The compact $(k,l)$-median rule}
\label{alg:compact}

\noindent
\textbf{Preconditions:} 
\begin{itemize}
\item $S_i$: \textbf{current} state of server $i$ (initially start state $s_0$ and seed command $x_0$) 
\item $L_i$: a log of non-committed commands together with their round numbers 
\item $sn(c)$: sequence number of  client $c$ (assumed to be included in $S_i$) 
\end{itemize}

\noindent\textbf{In every round, every server i does the following:}
\begin{itemize}
\item for every command $x$ received from a client $c$ with sequence number $sn(x)$: 
\begin{itemize}
    \item if $L_i \not= \bot$, $sn(x)=sn(c)+1$ and $x \not\in L_i$ then $i$ sends $\sigma \log n$ append requests with $x$ and current round number to servers chosen uniformly and independently at random
    \item if $L_i \not = \bot$ and $sn(x)=sn(c)$ then $i$ informs the client that a command with sequence number $sn(x)$ has already been committed 
    \item otherwise, $i$ ignores $x$
\end{itemize}
\item $i$ sends $k$ requests to servers chosen uniformly and independently at random and attaches a bit $b_i$, which is 1 if $L_i=\bot$, else 0
\item if $L_i \not= \bot$ then for any request received by $i$ from some server $j$, $i$ sends $L_i$ to $j$, and if $b_j=1$ then also $S_i$ back to $j$
\item if $i$ receives at least $\ell$ replies then it does the following:
  \begin{itemize}
  \item $i$ chooses a subset $M$ of $\ell$ of the replies uniformly at random 
  \item $i$ sets $L_i:=L'_i \circ \bar{L}$, where $L'_i$ is the median of the logs in $M$ (based on the lexicographical ordering) and $\bar{L}$ contains all commands not in $L'_i$ that are either contained in some log in $M$ or in an append request received in that round, in any order
    \item if the logs in $L_i$ contain different commands from the same client with the same sequence number then $i$ replaces all of them with $\bot$
  \item if $b_i=1$ then $i$ picks any reply $S_j$ from some $j$ in $M$ and updates its state $S_i := S_j$
  \item $i$ determines the largest prefix $P_i$ of $L_i$ where all commands have an age of at least $T$ 
  \item $i$ commits the commands in $P_i$ in the given order on $S_i$, removes $P_i$ from $L_i$, and increments the sequence numbers of the client for each command in $P_i$ \textit{(if $L_i$ becomes empty append a dummy command $x_d$)}
  \end{itemize}
\item if $i$ receives less than $\ell$ replies then it does the following:
  \begin{itemize}
  \item $i$ sets $L_i:=\bot$
  \item if at least one reply is received then $i$ picks any one of them and performs the updates on the shared space and sequence numbers of the clients as described above
  \end{itemize}
\end{itemize}
\end{algorithm}

We provide the formal proof that the compact $(6,3)$-median rule solves the commitment problem.

\begin{theorem}\label{th:compact}
If the adversary is 1-late and at least $n/4$ servers are useful in each round, the compact $(6,3)$-median rule solves the commitment problem (Definition~\ref{def:commitment_problem}) in $\bigO(\log n)$ rounds, w.h.p. Furthermore, given a 1-late $1/10$-blocking adversary, at least $3/4$ of the servers are useful for polynomially many rounds, w.h.p. The storage required for each server corresponds to the maximum number of commands injected within $O(\log n)$ rounds, the number of clients, and the maximum size of the shared state.
\end{theorem}

\begin{proof}
From Corollary \ref{cor:broadcast_command_median} we know that as long as the fraction of useful servers is at least $1/4$, it takes at most $T_B$ many rounds, w.h.p., until a command injected at a useful server has been broadcast to all useful servers, for some $T_B=\Theta(\log n)$. Thus, a command $x$ can only be accepted by a useful server if the first time it was accepted by a useful server is at most $T_B$ rounds ago. Therefore, it follows from the analysis of the $(6,3)$-median rule that after $T_B+T_M$ many rounds, all useful servers agree on the age of a command $x$ in their log, where $T_M$ is the time needed to reach a consensus in the $(6,3)$-median rule, w.h.p. If $T>T_B+\max\{T_B,T_M\}$, this implies that all useful servers will execute a command of age $T$ in their log in the same round, and therefore, strong safety is satisfied. Thus, Theorem~\ref{th:compact} follows from the analysis of the extended median rule.
\end{proof}

Of course, further optimizations can be done by sending only differentials, i.e., those parts of the shared state to server $j$ that are outdated in $j$, which can be achieved by memorizing the last positions in the sequence of committed commands that caused updates to certain parts of the shared state and remembering the last position of a command that contributed to the currently known shared state (not included in the protocol for brevity). 

\section{The Compact $(k,\ell)$-Median Rule with Certificates} \label{sec:proofs}

In typical blockchain applications, we want a client to be able to prove to other clients that a command issued by it has been committed. In cryptocurrencies, for example, that would allow a client to check whether a payment has been completed. 
However, in the previous section we allowed the servers to forget about already committed commands.
Certainly, the easiest way to provide a proof is that a server sends a signed message to the client when it tells the client that its command was committed.
However, signature keys can be stolen or become outdated, rendering previously signed messages useless.

A much safer way is to use a collision-resistant one-way hash function that is shared among the servers. Roughly speaking, the servers will store a small number of root hashes of a Merkle forest of the overall log of committed commands, whereas the clients will store appropriate hash chains of their committed commands. As we shall see, this allows the clients to prove to the servers that a command by that client was indeed committed, without having to know the whole log. Note that if a client loses the hash chain associated with some command $x$, the certification of commitment for $x$ is forfeit as well. We assume that it is in the interest of a client to safely store these hash chains. However, even in the event of loss of one such hash chain associated with a command $x$ does not affect the certification of the other commands, as the certification depends only on the servers keeping their hashes persistent.

We now describe the details of our construction. Suppose that $m$ commands have already been committed from the viewpoint of the useful servers (which have up-to-date information about the shared state and the log as described in the compact median rule). Let $m=2^{i_1}+2^{i_2}+\ldots+2^{i_a}$ be the unique decomposition of $m$ into powers of 2, where $i_1>i_2>\ldots>i_a$ and $a \in \mathbb{N}_0$. Consider the Merkle hash forest that consists of $a$ Merkle hash trees, where the $b$th Merkle hash tree, $b \in \{1,\ldots,a\}$, is a complete binary tree with $2^{i_b}$ leaves representing the commands at positions $\sum_{j<b} 2^{i_j} +1$ to $\sum_{j\le b} 2^{i_j}$ in the sequence of committed commands. We require that every useful server stores the following information about the Merkle hash forest and the clients:
\begin{itemize}
    \item The root hashes $h_1,\ldots,h_a$ of all Merkle hash trees.
    \item For each client, hash chains for the last two commands that got committed for that client together with their positions in the sequence of committed commands. For each such command $x$, the hash chain $hc(x)$ consists of the sibling hashes of $x$ and all ancestors of $x$ in its Merkle hash tree. 
\end{itemize}
Next, we describe the information stored by the clients. Suppose that the last two commands committed for some client $c$ are $x_1$ and $x_2$, where $x_1$ was committed earlier than $x_2$. Then $sn(x_1)=sn(c)-1$ and $sn(x_2)=sn(c)$, where $sn(c)$ is the largest sequence number committed for client $c$ from the viewpoint of the server. Recall that according to the compact median rule, a useful server that receives a command $x$ with $sn(x)=sn(c)$ from client $c$ informs $c$ that $x$ has been committed, which is the case here if $x=x_2$. In addition to that, the compact median rule with certificates will also send $hc(x_1)$ and $p(x_1)$ (the position of $x_1$ in the committed command sequence) to $c$. Each client will store all hash chains and positions received from the servers. It might happen that it receives hash chains of different length for the same command, in which case it keeps the longer one (as it is more up-to-date).

Suppose now that a client $c$ wants a server to verify that a command $x$ has been committed. Let $x_1,\ldots,x_k$ be the commands for which $c$ has already received hash chains from some servers and let $sn$ be the currently used sequence number of $c$. 
Certainly, for any $k \ge 1$, $sn = k+2$ since the largest sequence number for which $c$ was informed that the corresponding command got committed is $k+1$. Suppose that $c$ wants the servers to acknowledge that command $x$ with $sn(x) \in \{1,\ldots,k+1\}$ has been committed. Then it sends the certificate $(x,p(x),\overline{hc(x)})$ to the servers, where $\overline{hc(x)}$ is the longest hash chain for $x$ that $c$ can construct from the hash chains 
$hc(x_1),\ldots,hc(x_k)$. We can show that this certificate is indeed sufficient for the servers to verify that $x$ has been committed.

\begin{theorem}
\label{the:committed}
For any $x \in \{x_1,\ldots,x_{k+1}\}$, $(x,p(x),\overline{hc(x)})$ will allow any useful server to verify that $x$ has indeed been committed.  
\end{theorem}
\begin{proof}
If $x=x_{k+1}$ then $c$ does not yet have any information about $p(x)$ and $\overline{hc(x)}$. However, in that case any useful server still stores $x$ as one of the two last commands committed for client $c$ so that the verification is trivial.

Thus, suppose that $x=x_j$ for some $j \le k$. Let $lca(x_j,x_{k+1})$ be the least common ancestor of $x_j$ and $x_{k+1}$ under the assumption that the overall number of committed commands is large enough so that $x_j$ and $x_{k+1}$ are leaves of the same Merkle hash tree. Furthermore, given a node $v$ in a Merkle hash tree, let $lc(v)$ be its left child (which is the child leading to commands with lower positions in the ordering of the committed commands). While $lca(x_j,x_{k+1})$ might not yet exist in the Merkle hash forest, $lc(lca(x_j,x_{k+1})$ does exist in a sense that the Merkle hash (sub)tree rooted at $lc(lca(x_j,x_{k+1})$ is completely occupied (since otherwise $x_{k+1}$ would have to be part of that subtree).

\begin{lemma} \label{lem:hash-chain}
For every $j \le k$, $\overline{hc(x_j)}$ contains all sibling hashes from $x_j$ required to compute the hash value of $lc(lca(x_j,x_{k+1}))$.
\end{lemma}
\begin{proof}
We prove the lemma by induction on the number $k$ of commands for which client $c$ has already received hash chains. For the base case $k=1$, note that $x_1$ must be a leaf in the subtree rooted at the left child of $lca(x_1,x_2)$ while $x_2$ is a leaf in the subtree rooted at the right child of $lca(x_1,x_2)$. Thus, when a useful server informs $c$ that $x_2$ has been committed, all hashes of the subtree of $lc(lca(x_1,x_2))$ are already known to the server, which means that $hc(x_1)$ contains all sibling hashes starting from $x_1$ that are needed to compute the hash of $lc(lca(x_1,x_2))$.

Suppose now that the lemma already holds for $k-1$. To prove it for $k$, consider any $j<k$. ($j=k$ is trivial because we can set $hc(x_j):=\overline{hc(x_j)}$.) We distinguished between two cases. If $lca(x_j,x_k)=lca(x_j,x_{k+1})$ then the induction follows trivially since $lc(lca(x_j,x_k)) )=lc(lca(x_j,x_{k+1}))$.
Thus, assume that $lca(x_j,x_k) \not= lca(x_j,x_{k+1})$, which implies that $lca(x_j,x_k)$ is a descendant of $lca(x_j,x_{k+1})$. If a useful server informs $c$ that $x_{k+1}$ has been committed, all hashes of the subtree rooted at $lc(lca(x_k,x_{k+1}))$ are already known to the server, which implies that $hc(x_k)$ contains all sibling hashes starting from $x_k$ that are needed to compute the hash of $lc(lca(x_k,x_{k+1}))$. Since the path from $x_k$ to $lca(x_k,x_{k+1})$ passes through $lca(x_j,x_k)$, the old hash chain $\overline{hc(x_j)}$ for $k-1$ can be extended with the help of $hc(x_k)$ so that it leads all the way to $lc(lca(x_k,x_{k+1}))=lc(lca(x_j,x_{k+1}))$, which proves the lemma.
\end{proof}

Lemma~\ref{lem:hash-chain} implies that the server receiving $c$'s certificate can compute $h(lc(lca(x_j,x_{k+1})))$. In order to show that a useful server can verify the correctness of that hash value, we distinguish between two cases. If $lca(x_j,x_{k+1})$ does not yet exist, which is the case if some leaves of $lca(x_j,x_{k+1})$ are still missing because the number of committed commands is too small, then $lc(lca(x_j,x_{k+1}))$ must be the root of a Merkle hash tree. Hence, the server can verify $h(lc(lca(x_j,x_{k+1})))$ by comparing it with the corresponding root hash. If $lca(x_j,x_{k+1})$ does exist, the server has already been able to compute all hashes in the tree rooted at $lca(x_j,x_{k+1})$, which implies that $hc(x_{k+1})$ stored in the server contains all sibling hashes from $x_{k+1}$ till the right child of $lca(x_j,x_{k+1})$. Since the sibling hash of the right child of $lca(x_j,x_{k+1})$ is equal to $h(lc(lca(x_j,x_{k+1})))$, the server can verify the correctness of the certificate in this case as well, which proves the theorem.
\end{proof}

Note that the correctness of our certification scheme depends on the fact that the hash chain of $x_{k+1}$ is still stored in useful servers. This might not be the case if $\overline{hc(x)}$ in the certificate of $x$ becomes outdated, i.e., further commands of client $c$ get committed before the certificate for $x$ arrives at a useful server. This can be prevented if the client either waits with new commands until its request to verify the commitment of $x$ has been acknowledged by a server or it piggybacks such verification messages with requests to commit new commands.

\section{Recovery} \label{sec:recovery}

We now introduce a \textit{recovery protocol} to provide resilience against an adversary capable of blocking an arbitrary number of servers for an unlimited duration (i.e., the recovery problem; see Definition \ref{def:recovery_problem}). We call any period with less than $n/4$ useful servers a \textit{surge}. Any period where at least $n/4$ servers are useful is considered \textit{benign}. While the network will deteriorate into a death spiral, i.e., all servers become useless, somewhat above $n/4$ servers (see Lemma \ref{lem:spiral}), we can still guarantee proper functioning of our prior protocols as long as $n/4$ servers are useful.
However, our previous protocols are clearly not designed to handle a situation where all servers become useless, as they will lose all information on non-committed commands.

Therefore, we require servers to maintain persistent information on their states and logs, while at the same time keeping the protocol \textit{compact}, i.e., allowing servers to eventually commit commands to their shared state and subsequently delete them.
The primary challenge is that during a surge the adversary could allow some servers to progress and commit commands while preventing others from learning these. If the adversary then reverses this situation and permits the previously blocked servers to progress, those servers might commit a different set of commands, in violation of Definition \ref{def:recovery_problem}. Our goal in this section is to ensure that the sequence of committed commands does not fork even under a surge and quickly return to a valid solution for the commitment problem (Definition \ref{def:commitment_problem}), once the adversary becomes benign again. Formally, we show the following theorem.

\begin{theorem}
    \label{thm:recovery_problem}
    The recovery protocol (Algorithm \ref{alg:recovery_protocol}) solves the recovery problem (Def.\ \ref{def:recovery_problem}) and returns to a solution for the commitment problem (Def.\ \ref{def:commitment_problem}) $\bigO(\log n)$ rounds after a surge.
\end{theorem}

\subsection{The Recovery Protocol}

We begin with a brief overview of the recovery protocol, followed by pseudo-code (Algorithm \ref{alg:recovery_protocol}).
In the recovery protocol, time is divided into periods of $T$ rounds, referred to as $T$-\textit{windows}. Whenever a $T$-window ends and the protocol behaves as expected, servers commit commands and create a new \textit{checkpoint} that reflects the current state of servers. 
A checkpoint $C_i$ of server $i$ is a triple $(S,P,W)$, where $S$ is the current state of the state machine, $P$ is a sequence of sufficiently old (but not yet committed) commands, and $W$ is the $T$-window index when the checkpoint was taken. The parameters $S$ and $P$ represent a valid state of server~$i$ at some point, while $W$ provides a priority ordering (where newer checkpoints take precedence).
We also maintain a reset variable $R_i \in \{ \texttt{reset}, \texttt{no-reset}, \bot\}$ for each server $i$, indicating whether $i$ intends to revert to a previous checkpoint. If $R_i = \texttt{no-reset}$, server~$i$ expects normal operation and a benign adversary. If $R_i = \bot$, server~$i$ believes a surge is currently in progress. If $R_i = \texttt{reset}$, server~$i$ considers the surge over and that a rollback to a prior checkpoint is needed.

The complete recovery protocol can be viewed as two separate protocols: the \textit{reset protocol} and the \textit{checkpoint protocol} that run alongside the extended median rule, and incorporates the functionality of the compact median rule with certificates. The \textit{reset protocol} is governed by each server’s reset state, $R_i$. \textit{At the end} of each $T$-window, if server $i$ has an empty log $L_i$ (from the extended median protocol), it concludes that it was blocked and sets $R_i = \texttt{reset}$. Otherwise, $R_i = \texttt{no-reset}$. 
\textit{During} a $T$-window, in every round, each server~$i$ queries the reset states $R_j$ of $k$ randomly chosen servers. If at least $\ell$ responses are received and any of them is \texttt{no-reset}, the server sets $R_i = \texttt{no-reset}$; if all responses are \texttt{reset}, it sets $R_i = \texttt{reset}$. If fewer than $\ell$ responses are received, $R_i$ becomes $\bot$, and a server in the $\bot$ state does not answer queries about its own state $R_i$.
This design ensures that if the adversary surges during a $T$-window, the system falls into a “death spiral” where all servers eventually switch to $R_i = \bot$. Otherwise, the servers converge on either \texttt{reset} or \texttt{no-reset}. Finally, if a server ends a $T$-window with $R_i = \texttt{reset}$, it rolls back to a previously agreed-upon checkpoint.

The \textit{checkpoint protocol} revolves around the triple $C_i$. Servers exchange checkpoints by issuing $k$ random queries, analogous to how they exchange logs $\textit{L}_i$ or states $R_i$. The key distinction is that checkpoints are prioritized by age, given by the third parameter $W$. If the newest received checkpoint $(S', P', W')$ has a higher number $W'$ than its own checkpoint, it immediately adopts it and updates its state $S_i := S'$ and its log $L_i := P'$, since the server is outdated. This procedure spreads the newest checkpoints (those with the largest $W$) within one $T$-window, assuming $T=O(\log n)$ is sufficiently large.

For conciseness, we assume in this section that
the functionality of the compact median protocol with certificates is subsumed in the state machine, in particular, the required data structure centered around client sequence numbers (Section \ref{sec:cmedian}) and the Merkle hash forests (Section \ref{sec:proofs}). 
The techniques shown in sections \ref{sec:cmedian} and \ref{sec:proofs} can be adapted to work in the recovery protocol. Consequently, we will assume here that client commands are unique. 

\begin{algorithm}[!htp]
\setstretch{0.98}
\caption{Recovery Protocol}
\label{alg:recovery_protocol}

\noindent
\textbf{Preconditions:} 
\begin{itemize}

    \item $S_i, L_i$: \textbf{current} state and log of server $i$ (initially start state $s_0$ and seed command $x_0$) 

    \item $R_i \in \{ \texttt{reset}, \texttt{no-reset}, \bot\}$: server $i$'s willingness to roll back (initially \texttt{no-reset})

    \item $C_i = (S,P,W)$: parameters we intend to rescue in case of a surge:
    \begin{itemize}
        \item $S$: last observed shared state (initially the start state $s_0$)
        \item $P$: sequence of ``pre-committed'' commands (initially $\bot$)
        \item $W$: number of $T$-window (initially 0)
    \end{itemize}
\end{itemize}

\noindent
\textbf{Each round during a $T$-window each server $i$ does:} \textit{(on top of extended median rule)}
\begin{itemize}
    \item Send $k$ requests to random servers \textit{(same as in the extended median protocol)}
    \item If a request was received from server $j$ and $R_i \neq \bot$, then
    \begin{itemize}
        \item send $C_i$ and $R_i$ to server $j$
    \end{itemize}
    \item If at least $\ell$ replies were received, then
    \begin{itemize}
        \item If \texttt{no-reset} was received in one reply, $R_i := \texttt{no-reset}$, else $R_i := \texttt{reset}$
        \item let $C' = (S',P',W')$ be the received checkpoint with largest $W'$ \textit{(break ties arbitrarily)}
        \item If $W' > W$
        \begin{itemize}
            \item set $S_i := S'$, $L_i := P'$ \textit{(bring server up to date)}
            \item set $C_i := C'$ 
            \textit{(adopt newer checkpoint)}
        \end{itemize}
    \end{itemize}
    \item If less than $\ell$ replies were received, set $R_i :=\bot $
\end{itemize}

\noindent
\textbf{Between $T$-windows do:}
\begin{itemize}
    \item If $R_i = \texttt{reset}$, then set $S_i := S$, $L_i := P$, where $C_i = (S,P,W)$  \textit{(rollback)}
    \item If $L_i \neq \bot$, then
    \begin{itemize}
        \item for current checkpoint $C_i = (S,P,W)$, commit all commands in $P$ on $S_i$ and remove the prefix $P$ from $L_i$ \textit{(if $L_i$ would become empty, append a dummy command $x_d$)}.   
        \item make new checkpoint $C_i := (S_i,P_i,W')$, where $S_i$ is the current state, $P_i$ is the longest prefix of $L_i$ of commands with age $\geq T$ and $W'$ is the next $T$-window number 
        \item set $R_i := \texttt{no-reset}$ 
    \end{itemize}
    \item If $L_i = \bot$, then $R_i := \texttt{reset}$
\end{itemize}
\end{algorithm}

\subsection{Properties of the Recovery Protocol}

Before proving our main properties, we introduce several key time intervals used to define the length $T$ of each $T$-window:

\begin{itemize}
    \item \textbf{$T_M$:} The number of rounds required to solve the stabilizing consensus problem. By Theorem~\ref{th:keys}, $T_M = \bigO(\log n)$ if at least $n/4$ servers are continuously useful or if the adversary remains benign.

    \item \textbf{$T_E$:} The number of rounds needed for a successfully injected command to finalize its position in every log under the extended $(6,3)$-median rule. By Theorem~\ref{th:sequence}, $T_E \in \bigO(\log n)$ when $n/4$ servers are useful.

    \item \textbf{$T_B$:} The broadcast time for a value. From Theorem~\ref{thm:gossip_rule}, Corollary~\ref{cor:broadcast_priority_rule}, and Corollary~\ref{cor:broadcast_command_median}, we know $T_B = \bigO(\log n)$ when $n/4$ servers are useful under gossip, median, or priority rules.

    \item \textbf{$T_P$:} The time to reach a single value in the priority rule. By Corollary~\ref{cor:broadcast_priority_rule}, $T_P = \bigO(\log n)$ when $n/4$ servers are useful.

    \item \textbf{$T_D$:} The duration of a ``death spiral,'' i.e., the time until every server becomes \emph{useless} when the adversary blocks many servers. Theorem~\ref{lem:spiral} shows that if fewer than $(1/3 - \varepsilon)n$ servers remain useful in a single round (for any constant $\varepsilon > 0$), a death spiral completes in $\bigO(\log \log n)$ rounds. Moreover, if at least $3n/10$ servers are continuously blocked, then by Lemma~\ref{lem:permanently_blocking}, the death spiral completes in $\bigO(\log n)$ rounds. In both cases, we have $T_D \in \bigO(\log n)$.
\end{itemize}

Let us start by showing that the \texttt{reset} and \texttt{no-reset} states cannot coexist among servers at the end of a $T$-window (i.e., immediately before executing the instructions between two $T$-windows). This guarantees that no subset of servers can believe the protocol is functioning normally while others consider a reset necessary.

\begin{lemma}       
\label{lem:reset_states_mutually_exclusive}
    At the end of a $T$-window there are no two servers $i$ and $j$ with $R_i = \texttt{reset}$ and $R_j = \texttt{no-reset}$, w.h.p. This holds for sufficiently large $T = \bigO(\log n)$ even if the adversary surges.
\end{lemma}

\begin{proof}
    Let $T > T_B + T_D = \bigO(\log n)$. Consider the case where the $T$-window $W$ is not good. Then, by Lemma \ref{lem:spiral}, at the end of $W$ all servers are $\bot$ which satisfies this Lemma.

    Now consider the case that $W$ was good. By our algorithm each server $i$ starts the $T$-window either in the state $R_i := \texttt{reset}$ or $R_i := \texttt{no-reset}$. In the reset protocol we apply the $(6,3)$-priority rule (Definition \ref{def:priority_rule}) with the two values \texttt{reset} and \texttt{no-reset}, where the latter has priority.   
    By the application of the priority rule due  there will be agreement on either $\texttt{reset}$ or \texttt{no-reset} among all useful servers, all others will be $\bot$ (Corollary \ref{cor:broadcast_priority_rule}). 
\end{proof}

We observe that the states $R_i$ and $L_i$ have the following connection.

\begin{lemma}
    \label{lem:undecided_implies_useless}
    For every server $i$ in each round, $R_i = \bot$ implies $L_i = \bot$ (i.e., server $i$ is useless).
\end{lemma}

\begin{proof}
    The conditions under which we set $R_i := \bot$ (receiving less than $\ell$ replies) are sufficient for setting $L_i := \bot$. Furthermore, whenever we set $L_i$ to some value $\neq \bot$, then $R_i \neq \bot$.
\end{proof}

Next, we clarify how the death spiral relates to the servers' reset states $R_i$. We define a \emph{good} $T$-window $W$ as one in which, during the first $T - T_D$ rounds, more than $(1/3 - \varepsilon)n$ servers have state $R_i \neq \bot$, for some arbitrarily small constant $\varepsilon > 0$. This threshold stems from Lemma~\ref{lem:spiral}, which shows that if fewer than $(1/3 - \varepsilon)n$ servers have state $\bot$, a death spiral occurs, causing all servers to be in state $\bot$ after $T_D$ rounds.
Intuitively, a good $T$-window allows us to achieve agreement on checkpoints and reset states in the checkpoint and reset protocols, respectively.

\begin{lemma}
    \label{lem:undecided_vs_good}
    Assume $T > T_D$, where $T_D \in \bigO(\log n)$ is the time until the death spiral completes. If at the end of a $T$-window $W$ there exists a server with $R_i \neq \bot$, then $W$ was good, w.h.p.
\end{lemma}

\begin{proof}
    We show the contrapositive of the claim. Assume that $W$ was \textit{not} good, i.e., there is a round in the first $T-T_D$ rounds of $W$ in which less than $(1/3-\varepsilon)n$ servers are useful, i.e. are in state $R_i \neq \bot$. By Lemma \ref{lem:spiral}, a death spiral starts in this round which completes after at most $T_D$ rounds, w.h.p., implying that all servers have $R_i = \bot$. 
\end{proof}

We now show the first benefit of a \emph{good} $T$-window: after such a window, all reset states of servers $i$ with $R_i \neq \bot$ are identical. Moreover, if the number of servers in state \texttt{no-reset} exceed a certain minimum occurrence, we can guarantee that $\texttt{no-reset}$ will be adopted by all servers with $R_i \neq \bot$.
In the following, when we speak of the end of a $T$-window, we mean after the last round of the $T$-window but before the steps between $T$-windows are executed. The beginning of a $T$-window is the point after the steps between $T$-windows were executed but before the first round of the $T$-window (see the recovery protocol in Algorithm \ref{alg:recovery_protocol}). 

\begin{lemma}\label{lem:good_window_consensus_recovery}
If a given $T$-window $W$ was good, then for sufficiently large $T = \bigO(\log n)$, the following statements hold w.h.p. The reset state $R_i$ is identical for all servers $i$ with $R_i \neq \bot$ at the end of $W$. If, at the beginning of $W$, at least $c \log n$ non-blocked servers have the state $\texttt{no-reset}$, then every server in $i$ with $R_i \neq \bot$ has state $R_i = \texttt{no-reset}$ at the end of $W$. If all servers $i$ with $R_i \neq \bot$ start $W$ with $R_i = \texttt{reset}$, then every server in $i$ with $R_i \neq \bot$ has state $R_i = \texttt{reset}$ at the end of $W$. 
\end{lemma}

\begin{proof}
    Assume $T \geq T_P + T_D \in \bigO(\log n)$ is sufficiently large. Since $W$ was good, in the first $T_P$ rounds of this $T$-window, we have at least $(1/3-\varepsilon)n$ servers with $R_i \neq \bot$.

    In the reset protocol, servers follow a simple $(6,3)$-priority rule (Definition \ref{def:priority_rule}) where they adopt \texttt{no-reset} if they received at least $\ell$ responses in their $k$-sample one of which is \texttt{no-reset} (priority value), else \texttt{reset}, and $\bot$ if less than $\ell$ responses were received (corresponding to the useless state). If at least $c \log n$ servers $i$ with $R_i \neq \bot$ start with \texttt{no-reset}, by Corollary~\ref{cor:broadcast_priority_rule}, there will be a consensus among servers $i \in U$ on $R_i$ within $T_P$ rounds as this is a good $T$-window. If not, there will be agreement on $R_i$ among all $i \in U$ by Lemma \ref{lem:reset_states_mutually_exclusive} and in particular if all servers $R_i \neq \bot$ start with $R_i = \texttt{reset}$, then $\texttt{reset}$ must necessarily be the agreed upon value.
\end{proof}

The next lemma is crucial, as it shows that whenever a subset of servers creates a new checkpoint, they \emph{must} already know the most recent preceding checkpoint, ensuring a causal sequence with no alternative ``forks.'' We use an important observation about $(k,\ell)$-rules from Section~\ref{sec:median}: an adversary that continuously blocks $3n/10$ servers forces a death spiral. 

Concretely, if a new checkpoint is taken at the end of a $T$-window $W$, then $W$ must be \emph{good}, so at least $(1/3-\varepsilon)n$ servers store the most recent checkpoint. If the adversary suppresses (continuously blocks) more than $(1/3 - \varepsilon)n - c \log n \geq 3n/10$ of these servers (so there are fewer than $c \log n$ of servers with up-to-date checkpoints left), it triggers another death spiral, preventing any new checkpoint at the end of $W$. If it suppresses less than that in just a single round, the most up-to-date checkpoint will necessarily spread in the checkpoint protocol.

Thus, the adversary has only two choices: either suppress the most recent checkpoint so that no new checkpoint is taken, or allow the latest checkpoint to spread to all those servers that will create a new one. In both cases, the checkpoints follow a causal order. We formalize this in the next lemma.

\begin{lemma}\label{lem:good_window_consensus_checkpoint}
If a given $T$-window $W$ was good, then at the end of a $T$-window, all servers $i$ with $R_i \neq \bot$ (which is a super-set of those servers that make new checkpoints) will locally store a checkpoint with largest window number among all checkpoints currently stored by any server, for sufficiently large $T = \bigO(\log n)$ w.h.p.
\end{lemma}

\begin{proof}
    We make a proof by induction over successive $T$-windows. Our hypothesis is that at the end of each $T$-window with number strictly smaller than $W$, all servers $i$ with $R_i \neq \bot$, and at least $(1/3-\varepsilon)n$ servers overall will have a checkpoint with largest window number $W'$ among all checkpoints currently in circulation.
    
    Initially \textit{all} servers start with checkpoint $(s_0, \bot, 0)$, so the claim is true for $W=1$. If all servers $i$ have $R_i = \bot$ at the end $W$ then all servers have $L_i = \bot$ (Lemma \ref{lem:undecided_implies_useless}), thus no server makes a new checkpoint (see pseudocode) and the induction hypothesis carries over from the previous $T$-window $W-1$ to $W$.
    Else, there is a server $i$ with $R_i \neq \bot$, thus $W$ was good. Then in the first $T-T_D$ rounds of this $T$-window, we have at least $(1/3-\varepsilon)n$ servers with $R_i \neq \bot$.
    
    If a server receives at least $\ell$ replies in a round and observes a checkpoint whose window number $W'$ exceeds its own, it always adopts one of these. Otherwise, it keeps its current checkpoint. This process exactly matches the behavior of $(6,3)$-priority rule (Definition \ref{def:priority_rule}). Consequently, by Corollary~\ref{cor:broadcast_priority_rule}, if at least $c \log n$ servers with $R_i \neq \bot$ start $W$ with a checkpoint whose window number $W'$ is largest, that checkpoint spreads to all servers with $R_i \neq \bot$ but at least $(1/3-\varepsilon)n$ many overall, within $T_B$ rounds. 

    It remains to show that the checkpoint with the largest window number $W'$ will spread in the good $T$-window $W$. By the induction hypothesis a checkpoint with largest window number $W'$ is locally stored by at least $(1/3-\varepsilon)n$ servers at the start of $W$.    
    Assume that the adversary suppresses more than $(1/3-\varepsilon)n - c \log n$ of these for at least $T-T_B-T_D$ rounds. Since $(1/3-\varepsilon)n - c \log n \geq 3n/10$ for sufficiently small $\varepsilon$ and sufficiently large $n$, by Lemma \ref{lem:permanently_blocking}, all servers will have $R_i = \bot$ at the end of $W$, a contradiction. That implies that in some round in $W$ at the beginning of $T$ there were at least $c \log n$ servers with $R_i \neq \bot$ that had a checkpoint with largest window number $W'$ and there are still $T_B$ rounds left in each of which $\geq (1/3-\varepsilon)n$ servers have $R_i \neq \bot$. Then by Lemma \ref{cor:broadcast_priority_rule}, this checkpoint will spread to $(1/3-\varepsilon)n$ servers and all servers that have $R_i \neq \bot$ at the end of $W$. This satisfies the induction hypothesis for window $W$.
\end{proof}

Similar to a good $T$-window, we define a \textit{happy} $T$-window as one where during the first $T - T_D$ rounds more than $(1/3-\varepsilon)n$ servers have $L_i \neq \bot$ (they are useful). With the same argument as in Lemma \ref{lem:undecided_vs_good}, if there is a server $i$ with $L_i \neq \bot$ at the end of some $T$-window, a death spiral could not have occurred in the first $T-T_D$ rounds of $W$. Furthermore, as $L_i \neq \bot$ implies $R_i \neq \bot$ due to Lemma \ref{lem:undecided_implies_useless}, a happy $T$-window must also be good.

\begin{corollary}
    \label{cor:useful_implies_happy}
    Assume $T> T_D$, where $T_D \in \bigO(\log n)$ is the time until the death spiral completes. If at the end of some $T$-window $W$ there is a server $i$ with $L_i \neq \bot$, then $W$ was happy. Furthermore, a happy $T$-window is always good.
\end{corollary}

A \emph{happy} $T$-window has immediate effects on the local states, logs and checkpoints that are in circulation. In particular, when a $T$-window is happy, all \emph{useful} servers have aligned local states, and their logs share the same prefixes for all sufficiently old commands by its end. Because only useful servers create new checkpoints, any checkpoints bearing the same window number are identical.

\begin{lemma}    \label{lem:states_logs_checkpoints_consistent}
    The following properties hold for some $T \in \bigO(\log n)$, w.h.p. At the end of a $T$-window, all servers $i$ with $L_i \neq 
    \bot$ have the same local state $S_i$, have the same longest prefix of $L_i$ consisting of commands that are at least $T$ rounds old and have identical checkpoints $C_i$.
\end{lemma}

\begin{proof}
    Let $W' \geq 1$ be the current $T$-window.
    We make an inductive argument that all servers $i$ with $L_i \neq 
    \bot$ have the same local state $S_i$, have the same longest prefix of $L_i$ consisting of commands that are at least $T$ and make the same checkpoint at the end of $W'$. Note that the latter implies that $C_i$ is the same for all servers $i$ at the end of $W'+1$ due to Lemma \ref{lem:good_window_consensus_checkpoint} (for technical reasons the claim about $C_i$ is just shifted by one window). Assume the claim holds for all $T$-windows $W \leq W'$. This is true for $W=0$, as each server starts in state $S_i = s_0$ with log $L_i = x_0$ and initializes the same checkpoint $(s_0, \bot, 0)$.
    
    By induction hypothesis, all servers with $L_i \neq \bot$ start $W'$ in the same state $S_i = S$.
    If at the end of $W'$ any server $i$ has $L_i \neq \bot$, this implies that this was a happy $T$-window due to Corollary \ref{cor:useful_implies_happy}. Assume $T \in \bigO(\log n)$ is sufficiently large. 
    Since $L_i \neq \bot$ implies $R_i \neq 
    \bot$ (Lemma \ref{lem:undecided_implies_useless}) we can apply Lemma \ref{lem:good_window_consensus_checkpoint} implying that servers $i$ with $L_i \neq \bot$ have agreed on a checkpoint with the same window number $W \leq W'$ at the end of $W'$. All checkpoints $(S,P,W)$ with window number $W$ must be identical by the induction hypothesis. Any server $i$ that is useful ($L_i \neq \bot$) at the end of $W'$ has at some point obtained $(S,P,W)$, and will immediately adopt the state $S_i := S$ and sets $L_i := P$, in case it was outdated (see pseudocode).
    
    Further, for sufficiently large $T$, all \textit{useful} servers have enough time to exchange their logs by the application of the extended median rule, thus Theorem \ref{th:sequence} implies that all useful servers $i$ have the same longest prefix $P' \leq L_i$ of commands that are at least $T$ rounds old at the end of $W'$.
    Since $L_i \neq \bot$    
    is a necessary condition to commit commands and subsequently take a new checkpoint (see pseudocode), any server that commits commands, commits the same sequence $P$ to the same state $S$ from the consensus checkpoint $(S,P,W)$, resulting in the same new state $S'$. This further implies that all newly taken checkpoints at the end of window $W'$ are $(S',P',W')$, i.e., the same.
\end{proof}

\subsection{Liveness, Safety and Recovery}

Let us first demonstrate liveness by showing that the recovery protocol solves the commitment problem (Definition \ref{def:commitment_problem}), ensuring that new commands are committed as long as the adversary is benign, i.e., continuously admits at least $n/4$ useful servers. Critically, whenever a server becomes useful during a $T$-window, it adopts the most up-to-date checkpoint, preventing inconsistencies in committed commands. Furthermore, the adversary cannot exploit the \texttt{reset} state to continuously force resets making servers loose their logs (thus inhibiting the commitment of commands).

\begin{lemma}
    \label{lem:recovery_protocol_solves_commitment_problem}
    \label{lem:recovery_solves_commitment}
    The recovery protocol (Algorithm \ref{alg:recovery_protocol}) solves the commitment problem (Definition~\ref{def:commitment_problem}), given that the adversary is $1/10$-blocking.
\end{lemma}

\begin{proof}
    We consider a server $i$ as up-to-date in the sense of Definition \ref{def:commitment_problem} if $L_i \neq \bot$ (useful). Further, we consider a command as committed when it is removed from $P$ between two $T$-windows and executed on the shared state, by each useful server. Therefore, the sequence of committed commands of server $i$ is given implicitly via the current state $S_i$, which is a result of executing this sequence to the initial state $s_0$. By Lemma \ref{lem:available}
    \begin{itemize}
        \item \textbf{Strong Safety:} By Lemma \ref{lem:states_logs_checkpoints_consistent}, any useful server has the same state $S_i$ in each round, thus all useful servers agree on the sequence of committed commands.
        \item \textbf{Availability}: Holds due to the fact that if the adversary is benign, at least $n/4$ servers are useful.
        \item \textbf{Liveness:} As long as $n/4$ are useful, by Lemma \ref{lem:good_window_consensus_checkpoint}, all useful servers end (and thus also start) any $T$-window in the \texttt{no-reset} state, so they will \textit{not} do a rollback and therefore the logs $L_i$ stay intact (see pseudocode). Any command that was received by a server is spread to the logs $L_i$ of all useful servers $i$ in the extended median rule protocol (see Corollary \ref{cor:broadcast_command_median} and Lemma \ref{lem:new_command_stable}). At the latest at the end of the next $T$-window after a command was received, this command will be added to $P$ of the checkpoints of servers (``pre-committed''). One $T$-window after that the command will be executed on the shared state by all useful servers (see pseudocode).\qedhere
    \end{itemize}
\end{proof}

Next, we examine how our protocol preserves the \emph{monotonicity} condition of the {recovery problem} (Definition~\ref{def:recovery_problem}). Recall that monotonicity ensures each server’s sequence of committed commands never retracts previously committed entries to form an alternate sequence, even in a surge. Here, we show that if one server extends its committed sequence, other servers will only extend their own committed sequences in the same way. In particular, the process of sharing checkpoints and adopting a state of a checkpoint whenever a server realizes it is outdated, ensures that no server strays from the committed sequence of commands.

\begin{lemma}
    \label{lem:monotonicity}
    The recovery protocol (Algorithm \ref{alg:recovery_protocol}) ensures the monotonicity property of the recovery problem (Definition \ref{def:recovery_problem}), with an arbitrary blocking adversary.
\end{lemma}

\begin{proof}
    In this proof we associate a state $S_i$ of a server $i$ with the sequence of commands that transitions $s_0$ to the current state $S_i$. We say that some state $S$ is an extension of $S'$, if the sequence of commands associated with $S'$ is a prefix of those associated with $S$.
     
    We the make use of Lemma \ref{lem:states_logs_checkpoints_consistent}, which tells us that states $S_i$ 
    and checkpoints $C_i$, are identical among all useful servers $i$. 
    For this reason, each $T$-window $W$ can be associated with a state $S_W$ and sequence of pre-committed commands $P_W$ that is contained in the most recent checkpoint $(S_W,P_W,W')$ (i.e., largest $W'$) with $W' \leq W$ that any server holds during window $W$ (note that $W' \leq W$ since in case of a surge it is possible that no new checkpoints were made at the start of $W$).
    
    Our proof strategy is to show the following hypothesis inductively for successive window numbers $W$: for all $W'\leq W$ the state $S_{W}$ is an extension of $S_{W'}$.
    Each server starts window $W = 0$ with the start state $s_0$ and checkpoint $(s_0,\bot,0)$. This means that at the end of $W=0$, no commitments are made yet (as there are no pre-committed commands), thus all servers are still in state $s_0$ in window $W=1$. This shows the hypothesis for $W=1$. 
    
    It remains to show the claim for the next window $W+1$. By our hypothesis we have that $S_W$ is an extension of each $S_{W'}$ with $W' \leq W$. Let $(S_W,P_W,W')$ be the most up-to-date checkpoint in window $W+1$, which any useful server has by Lemma \ref{lem:states_logs_checkpoints_consistent}. Since only these useful server commit commands, they all commit $P_W$ to $S_W$, to obtain the same state $S_{W+1}$, which is therefore an extension of $S_W$.    
\end{proof}

It remains to show the recovery property of Definition \ref{def:recovery_problem}. The main point is that any surge will trigger a death spiral, where no server is useful anymore. This implies that if the surge ends, a reset will be triggered and all servers become useful again. From this point on all servers will share the most up to date checkpoints, such that eventually, all useful servers are up-to-date again.

\begin{lemma}
    \label{lem:recovery}
    The recovery protocol (Algorithm \ref{alg:recovery_protocol}) ensures the recovery property of the recovery problem (Definition \ref{def:recovery_problem}). Specifically, a valid solution of the commitment problem is established $\bigO(\log n)$ rounds after a surge ends.
\end{lemma}

\begin{proof}
    We have already seen that the commitment problem is solved if the adversary is benign in Lemma \ref{lem:recovery_protocol_solves_commitment_problem}. Assume that the adversary surges. This implies that a death spiral will occur by Lemma \ref{lem:spiral} and all servers become useless after $T_D$ rounds. 
    
    Since all useless servers $i$ will start a $T$-window in the state $R_i = \texttt{reset}$ (see pseudocode), the network will converge to $R_i = \texttt{reset}$ when at least $n/4$ servers are useful again for a complete $T$-window (see Lemma \ref{lem:good_window_consensus_recovery}). 
    In particular, these servers will have the most up-to-date checkpoint by Lemma \ref{lem:good_window_consensus_checkpoint} and Lemma \ref{lem:states_logs_checkpoints_consistent}.
    This implies that at least $n/4$ servers will trigger a reset and revert to this most up-to-date checkpoint that satisfies the monotonicity property (Lemma \ref{lem:monotonicity}) and start the next $T$-window in the state \texttt{no-reset}. From there on the system will proceed normally according to Lemma \ref{lem:recovery_protocol_solves_commitment_problem} until the adversary triggers the next death spiral. Overall the system will recover at most $3T \in \bigO(\log n)$ rounds after a surge ends.
\end{proof}

In summary, Lemmas \ref{lem:monotonicity} and \ref{lem:recovery} show the monotonicity and recovery property of the recovery problem (Definition \ref{def:recovery_problem}) which proves Theorem \ref{thm:recovery_problem}.

\section{Conclusion} \label{sec:extension}

We introduced a lightweight solution for state machine replication under an adversary capable of blocking any number servers. As shown in \cite{DoerrGMSS11}, the median rule can tolerate $\bigO(\sqrt{n})$ adversarial servers, suggesting that our approach can be extended to handle not just blocked but also $\bigO(\sqrt{n})$ malicious servers. Furthermore, we anticipate that it can accommodate server churn by allowing the adversary to have a $\beta$-fraction of blocked, joining or leaving servers in each round, thus removing the requirement of a static server set. 

Our $1/10$ bound for blocked servers to ensure liveness is not fundamental. By tuning parameters $k,\ell$, the protocol tolerates any 1-late $(1-\varepsilon)$-blocking adversary for constant $\varepsilon>0$. 
However, our choice of parameters provides some measure of partition tolerance.
With our chosen parameters, safety even holds under \emph{any} partition as all but at most one component will ``spiral to death'' (see Lemma \ref{lem:permanently_blocking}).
This provides a trade-off in the sense of the CAP theorem between safety, liveness, and partition tolerance \cite{GilbertL02}.

Moreover, we believe our protocols can be extended so that commitment errors are not persistent. First of all, we ensure that commitment errors are very unlikely and, therefore, rare in the first place. 
Second, the proof of Theorem~\ref{th:sequence} implies that the extended median rule achieves agreement on the position of a command $x$ within $O(\log n)$ rounds, w.h.p., \emph{irrespective} of how the logs look like initially, i.e., commitment errors only have limited effects on other commands. 
Third, our recovery protocol only takes another checkpoint of a shared state every $T=\Theta(\log n)$ rounds so that there is sufficient time to arrive at a consensus checkpoint that resolves disagreement on the shared state caused by a commitment error, by making use of the median rule. This pushes the probability of a \textit{persistent} error to a point where it can be truly ignored.

\appendix

\section{Probabilistic Concepts}
\label{sec:probabilistic_concepts}

In this section, we are going to cover the probabilistic concepts and inequalities that we use frequently throughout this paper.
The following two inequalities will turn out to be useful (see \cite{McDiarmid1998}).

\begin{lemma} \label{chernoff1}
Let $X_1,\ldots,X_n$ be independent binary random variables and let $X=\sum_{i=1}^n X_i$. For any $\mu \ge \E[X]$ and $\delta \ge 0$ it holds that
\begin{align*}
  \Pr[X \ge (1+\delta)\mu] & \le 
  \left( \frac{e^{\delta}}{(1+\delta)^{1+\delta}} \right)^{\mu} \\
  & \le  e^{-\min[\delta^2, \; \delta] \cdot \mu/3} .
\end{align*}
Furthermore, for any $0 \le \mu \le \E[X]$ and $0\le \delta \le 1$ it holds that
\begin{align*}  
  \Pr[X \le (1-\delta)\mu] & \le 
  \left( \frac{e^{-\delta}}{(1-\delta)^{1-\delta}} \right)^{\mu} \\
  & \le  e^{-\delta^2 \mu/2} .
\end{align*}
\end{lemma}

\begin{lemma} \label{chernoff2}
Let $X_1,\ldots,X_n$ be independent random variables with $a_k \le X_k \le b_k$ for all $k$, for some constants $a_k,b_k$, and let $X=\sum_{i=1}^n X_i$ and $\mu=\E[X]$. For any $t \ge 0$ it holds that
\[
  \Pr[|X-\mu| \ge t] \le 2 e^{-2t^2/\sum_k (b_k-a_k)^2} .
\]
\end{lemma}

\begin{lemma}[Union Bound for Events Holding w.h.p.]
    \label{lem:unionbound_whp}
    Let $E_1, \ldots, E_k$ be events each of which occurs w.h.p.\. Further, assume $k$ is polynomially bounded in $n$, i.e., $k \leq n^d$ for all $n \geq n_0$ for some monomial with fixed degree $d>0$.  Then $E := \bigcap_{i=1}^{k} E_i$ occurs w.h.p.
\end{lemma}

\begin{proof}
	Let $c>0$ be arbitrary but fixed. As $E_1, \dots , E_k$ occur w.h.p.\ we obtain the following. Let $c':= c+d$, there exist $n_1, \ldots , n_k \in \mathbb{N}$ such that for all $i \in \{1, \ldots, k\}$ we have $\mathbb{P}(\overline{E_i}) \leq \frac{1}{n^{c'}}$ for any $n > n_i$. With Boole's inequality we obtain
	\begin{align*}
		\mathbb{P}\big(\overline{E}\big) = 
		\mathbb{P}\Big(\bigcup_{i=1}^{k} \overline{E_i} \Big) 
		\leq \sum_{i=1}^{k} \mathbb{P}(\overline{E_i}) 
		\leq \sum_{i=1}^{k} \frac{1}{n^{c'}} 
		\leq  n^{d - c'}
		= n^{-c}
	\end{align*}
	for any $n \geq \max(n_0, \ldots ,n_k)$. 
\end{proof}

\begin{lemma}[Central Limit Theorem]
\label{lem:central_limit}
Let 
\[
    \Phi(x) = \frac{1}{\sqrt{2\pi}}\int_{-\infty}^{x} e^{-u^2/2} du
\]
and let $X=\sum_{i=1}^n X_i$ be a sum of independent random variables with finite $\mu = \E[X]$ and $\nu = \V[X]$. For any $a<b$ it holds that
\[
  \lim_{n \rightarrow \infty} \Pr \left[ a < \frac{X-\mu}{\sqrt{\nu}} < b \right] = \Phi(b)-\Phi(a).
\] 
\end{lemma}

Finally, we state a technical lemma from from \cite{DoerrGMSS11} (Claim 3.4), needed to show that in a $(k,\ell)$-gossip-rule, if only a single server initially has a broadcast value, then within $O(\log n)$ rounds, it will either spread to at least $c \log n$ servers or no server holds that value anymore.

\begin{lemma}\label{lem:dichotomy}
Let $(X_t)_{t=1}^{\infty}$ be a Markov Chain with state space $\{0, \ldots,
q\}$ that has the following properties:
\begin{itemize}
 \item there are constants $c_1 > 1$ and $c_2>0$, such that for any $t \in \mathbb N$,
$ \Pr[X_{t+1} \geq \min\{c_1 X_{t},q\} ] \geq 1 - e^{-c_2 X_t}$,
 \item $X_t=0 \Rightarrow X_{t+1} =0$ with probability $1$,
 \item $X_t=q \Rightarrow X_{t+1}=q$ with probability $1$.
\end{itemize}
Let $c_4 > 0$ be an arbitrary constant and $T:= \min \{ t \in \mathbb{N} \mid X_t
\in \{ \{ 0\} \cup \{q\} \}$. Then for every $X_1 \in \{0, \ldots,q\}$ and every constant $c_6 > 0$ there is a constant $c_5 > 0$ such that $\Pr[ T \leq c_5 \cdot \log q] \geq 1-q^{-c_6}$.
\end{lemma}

\bibliographystyle{ACM-Reference-Format}
\bibliography{bibliography}

\appendix

\end{document}